\documentclass[
  aps,
  prl,
  reprint,
  superscriptaddress,
  longbibliography,
  nofootinbib
]{revtex4-2}

\usepackage[T1]{fontenc}
\usepackage{amsmath,amssymb,amsthm,mathtools}
\usepackage{bm}
\usepackage{booktabs}
\usepackage{diagbox}
\usepackage{microtype}
\usepackage[hidelinks]{hyperref}
\usepackage{graphicx}
\usepackage{xcolor}
\usepackage{tikz}
\usetikzlibrary{arrows.meta,positioning}

\hypersetup{
  pdftitle={Bipartite Bound Information Exists},
  pdfauthor={Jef Pauwels, Nicolas Gisin, and Renato Renner},
  pdfsubject={Classical secret-key agreement and bound information}
}

\newcommand{\indep}{\mathrel{\perp\!\!\!\perp}}
\newcommand{\Fthree}{\mathbb F_3}
\newcommand{\Bern}{\operatorname{Bern}}
\newcommand{\BSC}{\operatorname{BSC}}
\newcommand{\BEC}{\operatorname{BEC}}

\newtheorem{theorem}{Theorem}
\newtheorem{proposition}[theorem]{Proposition}
\newtheorem{lemma}[theorem]{Lemma}
\newtheorem{corollary}[theorem]{Corollary}
\theoremstyle{definition}
\newtheorem{definition}[theorem]{Definition}
\newtheorem{remark}[theorem]{Remark}

\begin{document}

\title{Bipartite Bound Information Exists}

\author{Jef Pauwels}
\affiliation{Department of Applied Physics, University of Geneva, 1211 Geneva, Switzerland}
\affiliation{Constructor University, Bremen, Germany}
\author{Nicolas Gisin}
\affiliation{Department of Applied Physics, University of Geneva, 1211 Geneva, Switzerland}
\affiliation{Constructor University, Bremen, Germany}
\author{Renato Renner}
\affiliation{Institute for Theoretical Physics, ETH Zurich,
8093 Zurich, Switzerland}

\begin{abstract}
There exist bipartite quantum states that cost entanglement to create,
yet from which no singlet can be distilled.  This extreme
irreversibility is known as bound entanglement.  A quarter-century ago, Gisin and Wolf asked whether classical
information theory admits the same phenomenon.  Are there correlations, shared
by two parties and an eavesdropper, that cost secret bits to create, yet
from which none can be distilled?  This question is part of a broader
program exploring the relation between quantum and classical
information theory.  We answer it in the affirmative, providing a
simple example, a distribution of two bits and a trit.  Bound
entanglement thus has a classical counterpart, bound information.
An even more direct correspondence was originally conjectured, namely
that measuring purifications of bound-entangled states yields
distributions with bound information.  The very states that motivated the conjecture, however,
yield no bound information when measured in the standard basis, whereas
suitable measurements of purifications of separable states do.
The analogy between bound entanglement and bound information
thus lies in the accounting of resources, not in a correspondence between
individual quantum states and the classical probability distributions
induced by measuring them.
\end{abstract}

\maketitle

\begin{figure*}[t]
\centering
\begin{tikzpicture}[
  >=Latex,
  every node/.style={font=\footnotesize},
  resource/.style={draw, rounded corners=2pt, align=center,
    minimum width=3.10cm, minimum height=1.15cm, inner sep=4pt},
  bit/.style={draw, circle, align=center, minimum size=9mm, inner sep=1pt},
  channel/.style={-{Latex[length=2mm]}, semithick}
]
  \node[font=\bfseries] at (-4.60,1.50) {(a) Quantum-classical analogy};
  \node[resource, fill=blue!7] (quantum) at (-6.55,0)
    {\textbf{bound entanglement}\\[1mm]$E_{\rm C}>0,\quad E_{\rm D}=0$};
  \node[resource, fill=orange!10] (classical) at (-2.65,0)
    {\textbf{bound information}\\[1mm]$I_{\rm form}>0,\quad S=0$};
  \draw[<->, densely dashed, semithick, gray!75,
    shorten <=3pt, shorten >=3pt]
    (quantum.east) -- (classical.west);

  \draw[densely dashed, gray] (-0.35,-1.48) -- (-0.35,1.55);

  \node[font=\bfseries] at (3.85,1.50) {(b) Proof idea};
  \node[resource, minimum width=1.35cm, minimum height=0.78cm,
    fill=gray!8] (T) at (0.45,0) {$(X,Y)$};
  \node[bit, fill=gray!8] (L) at (1.85,0) {$L$};
  \node[resource, minimum width=2.55cm, minimum height=0.78cm,
    fill=blue!7] (Z) at (4.00,0.82)
    {$Z$: erase $L$ w.p.\ $\tfrac12$};
  \node[resource, minimum width=2.55cm, minimum height=0.78cm,
    fill=orange!10] (J) at (4.00,-0.82)
    {$J$: flip $L$ w.p.\ $\tfrac15$};
  \draw[channel] (T) -- (L);
  \draw[channel] (L) -- (Z.west);
  \draw[channel] (L) -- (J.west);
  \node[anchor=west, align=left, inner sep=0pt, font=\scriptsize]
    at (5.55,0.82)
    {$I(U;Z)\ge I(U;J)\;\forall\,U$\\[0.6mm]
     $\Rightarrow\;S=0$};
  \node[anchor=west, align=left, inner sep=0pt, font=\scriptsize]
    at (5.55,-0.82)
    {decoupling reproduces $J$,\\[0.6mm]
     no channel $Z\to J$\\[0.6mm]
     $\Rightarrow\;I{\downarrow}>0$};
  \node[below=2.2mm of J, font=\scriptsize] {$X\indep Y\mid J$};
\end{tikzpicture}
\caption{\textbf{Quantum-classical analogy and proof idea.}
(a) Bound entanglement and bound information express the same operational
irreversibility, namely a positive creation cost (entanglement cost $E_{\rm C}$,
formation cost $I_{\rm form}$) with zero distillable yield (distillable
entanglement $E_{\rm D}$, secret-key rate $S$).
(b) In our source (Table~\ref{tab:source}), a latent bit $L$ is correlated
with $(X,Y)$ through Eq.~\eqref{eq:Lchannel} and passed through two noisy
channels.  Eve's observation $Z$ reveals $L$ half of the time and is
blank otherwise.  The extra variable $J$ is $L$ flipped with
probability $1/5$ and decouples $X$ from $Y$.  Eve's channel is
less noisy than that of $J$, so no key can be distilled.  Every processing of $Z$ that decouples $X$ from $Y$ must reproduce $J$
up to a merging of outputs, and no channel $Z\to J$ exists, so the
intrinsic information stays positive.}
\label{fig:mechanism}
\end{figure*}
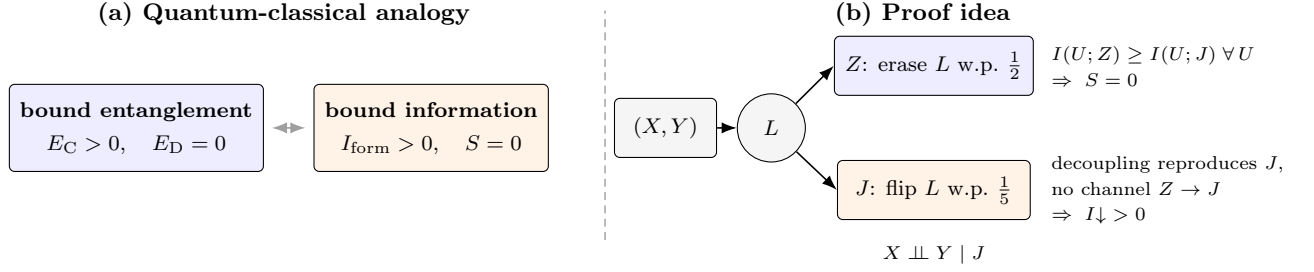

\emph{Introduction.---}%
Schr{\"o}dinger singled out entanglement as \emph{the} characteristic
trait of quantum theory~\cite{Schroedinger1935}.  Quantum communication
turned it into a currency, spent in teleportation and entanglement-based
cryptography~\cite{BennettTeleportation1993,Ekert1991}.  A currency
invites bookkeeping.  How much entanglement does a state cost to create,
and how much can be recovered from it?  It came as a surprise that the
books can fail to balance in the most extreme
way~\cite{Horodecki1998,HorodeckiReview2009}.  There exist \emph{bound-entangled} states, i.e., bipartite states that cannot be created without entanglement, yet from
which no singlet can be distilled.

Classical cryptography has a currency of its own, namely secret correlations.
Alice, Bob, and the adversary Eve receive independent repetitions of a
triple $(X,Y,Z)$ of random variables, emitted by a source
with a known joint distribution over finite alphabets.  Alice sees only $X$, Bob only $Y$, and Eve only $Z$.
Using local randomness and unlimited authenticated public discussion,
which Eve reads in full, Alice and Bob attempt to agree on a secret key, i.e., shared uniform
bits about which Eve has essentially no information.  The maximal number of such bits that can be
generated per repetition of the triple is the secret-key rate
$S(X;Y\Vert Z)$~\cite{Maurer1993,AhlswedeCsiszar1993}
(for the equivalent definitions of
secrecy~\cite{MaurerWolf2000,RennerThesis2005,BenOrEtAl2005,PortmannRenner2022},
see Appendix~\ref{sec:conventions}).
This source model describes the classical post-processing stage of
quantum key distribution when Eve's side information is classical.
Alice and Bob distill their key from measured data by public
discussion~\cite{GisinEtAlQKD2002}.

Entanglement and secret correlations are connected
through a purely classical quantity, the intrinsic
information $I(X;Y{\downarrow}Z)$.  It was originally introduced by
Maurer and Wolf as an upper bound on the secret-key
rate~\cite{MaurerWolf1999}.  To relate $I(X;Y{\downarrow}Z)$ to
entanglement, let Alice and Bob share a bipartite quantum state and let
Eve hold a purifying system, which carries all that an adversary can
know about their state.  Local measurements by the three parties then
turn this pure tripartite state into a classical source.  Gisin and Wolf showed that the reduced state of Alice and Bob is
entangled if and only if, whatever measurement Eve performs, Alice and
Bob can choose measurements whose outcome distribution has positive
intrinsic information~\cite{GisinWolf2000,GisinRennerWolf2001,ChristandlRennerWolf2003}.
Entanglement is thus visible in purely classical data.  The analogy
between entanglement and secret correlations extends well beyond
this criterion, as secret correlations mirror entanglement in a
whole dictionary of protocols~\cite{CollinsPopescu2002}, and classical
and quantum key distillation admit a unified
treatment~\cite{ChristandlEtAl2007}.

This analogy suggested that bound entanglement, too, should leave a
classical trace.  Gisin and Wolf conjectured the existence of \emph{bound
information}~\cite{GisinWolf2000}, that is, of a distribution whose intrinsic
information is positive although its secret-key rate vanishes.  More
specifically, they expected that measuring purifications of
bound-entangled states yields such distributions~\cite{GisinWolf2000}.  Renner and
Wolf later made the analogy to quantum theory
exact~\cite{RennerWolf2003}.  They introduced
the formation cost $I_{\rm form}(X;Y\mid Z)$, defined as the minimal rate of
pre-shared secret bits needed to generate samples of $(X,Y)$ by public
discussion while ensuring that the public transcript can be simulated
from Eve's target variable $Z$.  A source with
\begin{equation}
  S(X;Y\Vert Z)=0<I_{\rm form}(X;Y\mid Z)
  \label{eq:definition}
\end{equation}
therefore costs secrecy to create although none can ever be recovered.  This is
the defining property of bound entanglement, transcribed into classical
probability.

The conjecture that bound information exists remained open
for more than twenty-five years.  The difficulty lies in
the vanishing rate, since $S=0$ means that no
protocol whatsoever can distill a key, however many rounds of
interaction and however many copies of the source it uses.
Work in this direction yielded partial results.  Renner and
Wolf~\cite{RennerWolf2002,RennerWolf2003} constructed sources
whose formation cost exceeds their secret-key rate by an arbitrarily large
factor.  A source of rate exactly zero, however, remained out of reach.
Further work established non-distillability
criteria~\cite{MasanesWinter2010}, raised the
question whether individually useless sources can be
combined into a useful one~\cite{PretticoAcin2013}, and gathered
numerical evidence for bound secrecy under restricted two-way
post-processing~\cite{KhatriLutkenhaus2017}, but did not
settle the conjecture itself.
The problem changes character when the secret key is to be shared
among more than two honest parties.  In such a multipartite setting,
bound information does exist~\cite{AcinCiracMasanes2004,MasanesAcin2006}.  The proofs
certify non-distillability by grouping the parties across different
bipartitions, which is not possible in the bipartite setting of the present
work, where only two honest parties share the key.

The question of whether bipartite bound information exists belongs
to the broader program of exploring the relation between quantum and
classical information theory.  A central task of this program, from
the viewpoint of quantum foundations, is to clarify what makes quantum
theory genuinely quantum.  This requires knowing its closest classical
analogues.  Whether bound entanglement has such an analogue is
precisely the question of whether bound information exists.  Our main
contribution is to settle this question.  We also show, however, that the
correspondence between bound entanglement and bound information is not
the one originally conjectured.

Our first result is the explicit source of
Table~\ref{tab:source}, consisting of two bits and a trit satisfying
Eq.~\eqref{eq:definition}, which proves that bound information exists.
Our second result concerns the question whether there is a
direct relation between bound-entangled states and classical
distributions with bound information.  For this we consider the
Horodecki qutrit family~\cite{HorodeckiActivated1999}.  Gisin and Wolf's
conjecture grew out of the distributions obtained by measuring
these states in the standard bases~\cite{GisinWolf2000}.  We prove that these distributions are
key-distillable throughout the range in which their creation costs
any secrecy.  The distillability carries over from the distributions
to the Horodecki states themselves, so even the bound-entangled members
have positive distillable key, contrary to the conjectured correspondence.
We also show, in Appendix~\ref{sec:horodecki-bi-measurement}, that other measurements of purifications of the same states do yield bound
information and, perhaps surprisingly, that they do so for separable and
bound-entangled members alike.  There is thus
no direct relation in either direction.  A bound-entangled state need
not yield bound information when measured, and bound information in
measured data need not come from a bound-entangled state.

\emph{Proof strategy.---}%
We prove the first result through the intrinsic
information~\cite{MaurerWolf1999},
\begin{equation}
 I(X;Y{\downarrow}Z):=\inf_{Z\to\bar Z}I(X;Y\mid\bar Z),
 \label{eq:intrinsic}
\end{equation}
the smallest conditional mutual information Eve can reach by processing her
data through a channel $Z\to\bar Z$.  The
intrinsic information sits between the two operational
quantities~\cite{MaurerWolf1999,RennerWolf2003}:
\begin{equation}
 S(X;Y\Vert Z)\le I(X;Y{\downarrow}Z)\le I_{\rm form}(X;Y\mid Z).
 \label{eq:sandwich}
\end{equation}
It therefore suffices to establish $S=0<I{\downarrow}$.

To prove the two halves, $S=0$ and $I{\downarrow}>0$, we compare
Eve's observation $Z$ with another variable $J$.  There are two
different senses in which $Z$ can be at least as informative as $J$,
and they are not equivalent.  Eve \emph{simulates} $J$ if some
channel applied to $Z$ reproduces the joint distribution of $J$ with
$(X,Y)$~\cite{Blackwell1953}.  The variable $Z$ \emph{dominates} $J$ if it is more
informative in every mutual-information comparison, i.e., 
$I(U;Z)\ge I(U;J)$ for every joint input law $P_{UXY}$,
with the outputs generated through the fixed channels $P_{Z\mid XY}$
and $P_{J\mid XY}$.
Dominance in this sense is the \emph{less-noisy} order of K{\"o}rner and
Marton~\cite{KornerMarton1977}, formulated for the channels from $(X,Y)$
to $Z$ and to $J$.  We denote it by $P_{Z\mid XY}\succeq_{\rm ln}P_{J\mid XY}$.  Simulation implies dominance
by the data-processing inequality.  The converse is false.  The infimum in
Eq.~\eqref{eq:intrinsic} grants Eve only simulations.  A vanishing
secret-key rate, however, can be certified by the broader relation, as
shown by Gohari, G{\"u}nl{\"u}, and Kramer (GGK)~\cite{GohariGunluKramer2020}.
If Eve dominates a variable $J$ that renders $X$ and $Y$
conditionally independent, no protocol can distill a key.  Our existence proof of bound information relies on a source
that is situated in this gap [Fig.~\ref{fig:mechanism}(b)].

\begin{table}[b]
\caption{\label{tab:source}%
The bound-information source.  Each cell lists, for the corresponding pair
$(x,y)$, the values $z$ of Eve's variable in parentheses, followed by the
joint probability $P(X{=}x,Y{=}y,Z{=}z)$ in units of $1/36$.  Combinations
of probability zero are omitted.  For example,
$P(0,0,0)=\frac5{36}$ and $P(0,1,\bot)=\frac4{36}$.}
\begin{ruledtabular}
\begin{tabular}{c|cc}
 \diagbox{$y$ $(z)$}{$x$} & $0$ & $1$\\
\hline
 $0$ & $(0)\;5\quad(\bot)\;5$ & $(0)\;2\quad(1)\;2\quad(\bot)\;4$\\
 $1$ & $(0)\;2\quad(1)\;2\quad(\bot)\;4$ & $(1)\;5\quad(\bot)\;5$\\
\end{tabular}
\end{ruledtabular}
\end{table}

\emph{The source.---}%
We now present the source that realizes bound information.  Let
$X,Y\in\{0,1\}$ and $Z\in\{0,1,\bot\}$ be distributed according to
Table~\ref{tab:source}.
\begin{theorem}
\label{thm:bound-information}
The source of Table~\ref{tab:source} has bipartite bound information:
\begin{equation}
 S(X;Y\Vert Z)=0
 <I(X;Y{\downarrow}Z)
 \le I_{\rm form}(X;Y\mid Z).
 \label{eq:main-theorem}
\end{equation}
\end{theorem}
For the proof, we collect, for each value $z$, the unnormalized
\emph{slice} $M_z(x,y):=P(X{=}x,Y{=}y,Z{=}z)$, a $2\times2$ matrix that
can be read off the table.  For instance,
$M_\bot=\frac1{36}\left(\begin{smallmatrix}5&4\\4&5\end{smallmatrix}\right)$.

The construction is best understood through a latent bit $L\in\{0,1\}$,
generated from the pair $(X,Y)$ by the channel $P_{L\mid XY}$
defined by
\begin{equation}
 P(L{=}1\mid x,y)=
 \begin{pmatrix}0&1/2\\[1mm]1/2&1\end{pmatrix}_{x,y}.
 \label{eq:Lchannel}
\end{equation}
A direct check gives the joint slices $P_{XY,L=l}=2M_l$. In particular,
$L$ is uniform.  Eve's variable arises by sending $L$ through an erasure
channel, so that $Z$ equals $L$ with probability $1/2$ and the blank symbol $\bot$
otherwise.  This halves the two slices of $L$ and creates $M_\bot$, and so
reproduces Table~\ref{tab:source} exactly.

Now send the same latent bit through a flip channel and call the output
$J$, which equals $L$ with probability $4/5$ and $1-L$ otherwise.  The flip
probability $1/5$ is not arbitrary.  It is chosen precisely so that the
two slices of $J$ factorize:
\begin{equation}
\begin{split}
 A:=P_{XY,J=0}
 &=\frac1{18}\begin{pmatrix}4&2\\2&1\end{pmatrix}
 =\frac1{18}\begin{pmatrix}2\\1\end{pmatrix}
   \begin{pmatrix}2&1\end{pmatrix},\\
 B:=P_{XY,J=1}
 &=\frac1{18}\begin{pmatrix}1&2\\2&4\end{pmatrix}
 =\frac1{18}\begin{pmatrix}1\\2\end{pmatrix}
   \begin{pmatrix}1&2\end{pmatrix}.
 \label{eq:AB}
\end{split}
\end{equation}
Both have rank one.  Given $J$, the pair $(X,Y)$ is therefore a product
distribution. We write $X\indep Y\mid J$ and say that $J$ \emph{decouples}
$X$ from $Y$.  Note also $A+B=P_{XY}$.

With the latent bit $L$ and the decoupling variable $J$ in
place, we prove the two halves of
Theorem~\ref{thm:bound-information}, beginning with the vanishing rate.

\emph{Zero distillable key.---}%
Take any auxiliary variable $U$ and any joint law of $(U,X,Y)$, and
apply the channel of Eq.~\eqref{eq:Lchannel} followed by the erasure and
flip channels to obtain $L$, $Z$, and $J$.  Since the erasure and flip channels
act on $L$ with fresh randomness, $(Z,J)$ is conditionally independent
of $(U,X,Y)$ given $L$.  The mutual informations with $U$ then satisfy
\begin{equation}
 I(U;Z)=\tfrac12\,I(U;L),\qquad
 I(U;J)\le\tfrac9{25}\,I(U;L).
 \label{eq:contractions}
\end{equation}
The equality holds because $Z$ reveals $L$ exactly half of the time.  The
inequality is the strong data-processing inequality of the flip channel,
$I(U;J)\le(1-2\delta)^2I(U;L)$ at flip probability
$\delta$~\cite{AhlswedeGacs1976} (a self-contained proof is in
Appendix~\ref{sec:bsc}).
Since $\tfrac12>\tfrac9{25}$, $Z$ dominates $J$, that is,
$P_{Z\mid XY}\succeq_{\rm ln}P_{J\mid XY}$.

GGK proved that dominance carries over
to the secret-key rate, whatever interactive protocol is
used~\cite{GohariAnantharam2010,GohariGunluKramer2020}:
\begin{equation}
 P_{Z\mid XY}\succeq_{\rm ln}P_{J\mid XY}
 \ \Longrightarrow\
 S(X;Y\Vert Z)\le S(X;Y\Vert J).
 \label{eq:GGK}
\end{equation}
Appendix~\ref{sec:capacity-comparison} contains an independent proof.
Because $J$ decouples $X$ from $Y$, Eq.~\eqref{eq:sandwich} applied to $J$
gives $S(X;Y\Vert J)=0$, and therefore
\begin{equation}
 S(X;Y\Vert Z)=0.
 \label{eq:zero}
\end{equation}

\emph{Positive formation cost.---}%
The second half of the proof establishes the positive intrinsic
information $I(X;Y{\downarrow}Z)>0$.  By Eq.~\eqref{eq:intrinsic},
$I(X;Y{\downarrow}Z)=0$ means that $I(X;Y\mid\bar Z)$ can be made
arbitrarily small by processing $Z$.  For finite alphabets the infimum is
attained~\cite{ChristandlRennerWolf2003} (see Appendix~\ref{sec:finite-zero} for a compactness proof), so some
channel $Z\to\bar Z$ would achieve $I(X;Y\mid\bar Z)=0$, that is,
decouple $X$ from $Y$.  Our proof that no such channel exists has two
steps.  We first show that any channel $Z\to\bar Z$ that decouples $X$
from $Y$ must, after a merging of its outputs, reproduce the variable
$J$.  We then show that no channel $Z\to J$ exists.

For the first step, note that the source
slices lie in the plane spanned by the two product slices of
Eq.~\eqref{eq:AB},
\begin{equation}
 M_0=\tfrac23A-\tfrac16B,\quad
 M_1=-\tfrac16A+\tfrac23B,\quad
 M_\bot=\tfrac12(A+B),
 \label{eq:basis}
\end{equation}
so for any channel $Z\to\bar Z$, every slice $N_k=P_{XY,\bar Z=k}$ has the
form $N_k=a_kA+b_kB$.  A direct calculation yields, for every output $k$,
\begin{equation}
 \det(a_kA+b_kB)=\frac{a_kb_k}{36}.
 \label{eq:det}
\end{equation}
Suppose now that $\bar Z$ decouples $X$ from $Y$, that is,
$I(X;Y\mid\bar Z)=0$.  Conditional independence means that
every nonzero slice $N_k$ is a nonnegative matrix of rank one, so Eq.~\eqref{eq:det} forces $a_kb_k=0$.  Hence
each $N_k$ lies on the ray of $A$ or of $B$.  Merge the outputs
accordingly, mapping every $A$-type output to $0$ and every $B$-type
output to $1$.  The merged variable is again a processing of $Z$, and its
two slices are $\lambda A$ and $\mu B$ for some $\lambda,\mu\ge0$.  Slices
always sum to the marginal, so $\lambda A+\mu B=P_{XY}=A+B$.  Since $A$
and $B$ are linearly independent, $\lambda=\mu=1$, so the merged variable has
exactly the slices of $J$.

So far we have shown that any channel turning $Z$ into an exact decoupler
would simulate $J$.  It remains to show that no channel $Z\to J$ exists.
Suppose one did, and write $r_z=P(J{=}1\mid Z{=}z)$.  Given
$(X,Y)=(0,0)$, the latent bit is $L=0$, so $Z$ is $0$ or $\bot$ with
probability $\tfrac12$ each, while $J=1$ must occur with probability
$\tfrac15$.  Given $(1,1)$, the same reasoning gives $\tfrac45$:
\begin{equation}
 \frac{r_0+r_\bot}{2}=\frac15,\qquad
 \frac{r_1+r_\bot}{2}=\frac45.
 \label{eq:nodegrade}
\end{equation}
The two conditions are incompatible.  Subtracting the first from the second
eliminates $r_\bot$ and gives $r_1-r_0=6/5$, which is impossible for
probabilities.  We conclude
\begin{equation}
 I(X;Y{\downarrow}Z)>0.
 \label{eq:positive}
\end{equation}
Together with Eqs.~\eqref{eq:sandwich} and \eqref{eq:zero}, this proves
Theorem~\ref{thm:bound-information}.\footnote{By a recent result, the
reduced intrinsic information of a finite source vanishes exactly when the
intrinsic information does~\cite{KhesinTungVedula2023}.  The source of
Table~\ref{tab:source} therefore also satisfies
$S=0<I{\downarrow\downarrow}$, the improved bound of
Refs.~\cite{RennerSkripskyWolf2003,RennerWolf2003}.}
The example is not isolated.  Appendix~\ref{sec:family} embeds it into a two-parameter family
with the same properties.

\begin{table}[b]
\caption{\label{tab:gw}%
Standard-basis distributions of the Horodecki family, for $2\le\alpha\le5$, corresponding to Example~3 of
Refs.~\cite{GisinWolf2000,GisinRennerWolf2002} (see also
Ref.~\cite{RennerDiploma2000}).  Entries as in Table~\ref{tab:source}, in
units of $1/21$.  Eve's symbol is $\bot$ on the whole
diagonal, so she cannot tell which value $X=Y$ took, while each
off-diagonal outcome reveals the pair completely.
In the quantum picture, $Z$ labels the eigenstates of $\sigma_\alpha$
(Fig.~\ref{fig:alpha}); all diagonal outcomes stem from the entangled
eigenstate $(|00\rangle+|11\rangle+|22\rangle)/\sqrt3$.  The family is
symmetric under $\alpha\leftrightarrow5-\alpha$ with the roles of Alice
and Bob exchanged, so the range $2\le\alpha\le5$ covers all cases.}
\begin{ruledtabular}
\begin{tabular}{c|ccc}
 \diagbox{$y$ $(z)$}{$x$} & $0$ & $1$ & $2$\\
\hline
 $0$ & $(\bot)\;2$ & $(10)\;5-\alpha$ & $(20)\;\alpha$\\
 $1$ & $(01)\;\alpha$ & $(\bot)\;2$ & $(21)\;5-\alpha$\\
 $2$ & $(02)\;5-\alpha$ & $(12)\;\alpha$ & $(\bot)\;2$\\
\end{tabular}
\end{ruledtabular}
\end{table}

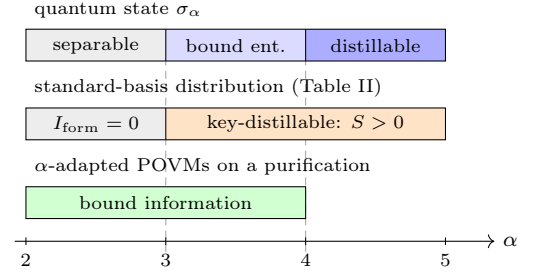
\begin{figure}[t]
\centering
\begin{tikzpicture}[xscale=1.85]
  \draw[densely dashed, gray!70] (3,-0.14) -- (3,2.80);
  \draw[densely dashed, gray!70] (4,-0.14) -- (4,2.80);
  \draw[->] (1.92,0) -- (5.35,0) node[right,font=\footnotesize]{$\alpha$};
  \foreach \x in {2,3,4,5}
    \draw (\x,0.05) -- (\x,-0.05) node[below,font=\scriptsize]{$\x$};

  \fill[green!18] (2,0.30) rectangle (4,0.72);
  \draw (2,0.30) rectangle (4,0.72);
  \node[font=\scriptsize] at (3,0.51) {bound information};
  \node[anchor=south west,font=\scriptsize] at (2,0.76)
    {$\alpha$-adapted POVMs on a purification};

  \fill[gray!14]   (2,1.34) rectangle (3,1.76);
  \fill[orange!22] (3,1.34) rectangle (5,1.76);
  \draw (2,1.34) rectangle (5,1.76);
  \draw (3,1.34) -- (3,1.76);
  \node[font=\scriptsize] at (2.5,1.55) {$I_{\rm form}=0$};
  \node[font=\scriptsize] at (4,1.55) {key-distillable: $S>0$};
  \node[anchor=south west,font=\scriptsize] at (2,1.80)
    {standard-basis distribution (Table~\ref{tab:gw})};

  \fill[gray!14]  (2,2.38) rectangle (3,2.80);
  \fill[blue!14]  (3,2.38) rectangle (4,2.80);
  \fill[blue!32]  (4,2.38) rectangle (5,2.80);
  \draw (2,2.38) rectangle (5,2.80);
  \draw (3,2.38) -- (3,2.80);
  \draw (4,2.38) -- (4,2.80);
  \node[font=\scriptsize] at (2.5,2.59) {separable};
  \node[font=\scriptsize] at (3.5,2.59) {bound ent.};
  \node[font=\scriptsize] at (4.5,2.59) {distillable};
  \node[anchor=south west,font=\scriptsize] at (2,2.84)
    {quantum state $\sigma_\alpha$};
\end{tikzpicture}
\caption{\textbf{One state family, two classicalizations.}
Top: the Horodecki qutrit states
$\sigma_\alpha=\frac27\,\Pi_\Phi+\frac\alpha7\,\sigma_+
 +\frac{5-\alpha}7\,\sigma_-$,
where $\Pi_\Phi$ projects onto
$(|00\rangle+|11\rangle+|22\rangle)/\sqrt3$ and $\sigma_\pm$ are the
uniform mixtures of $|i,i\pm1\rangle$ (indices modulo three), are separable
for $2\le\alpha\le3$, bound entangled for $3<\alpha\le4$, and distillable
for $4<\alpha\le5$~\cite{HorodeckiActivated1999}.
Middle: the standard-basis distributions of Table~\ref{tab:gw} have a
single transition at $\alpha=3$.  Creating them costs no secrecy for
$2\le\alpha\le3$, whereas a secret key is distillable for every
$3<\alpha\le5$ (Theorem~\ref{thm:gw}).
Bottom: the $\alpha$-adapted POVMs of
Appendix~\ref{sec:horodecki-bi-measurement},
applied to a purification of $\sigma_\alpha$, produce one and the same
bound-information distribution for every $2\le\alpha\le4$.  Thus the
identical classical resource arises from both separable and bound-entangled
members of the family.  This alternative construction is not claimed for
$4<\alpha\le5$.}
\label{fig:alpha}
\end{figure}

\emph{The quantum analogy revisited.---}%
Theorem~\ref{thm:bound-information} establishes the analogy at the level of
resources [Fig.~\ref{fig:mechanism}(a)].  We now turn to the
second result, the question whether there is a direct relation between
bound-entangled states and classical distributions with bound
information.  A first obstacle to such a relation lies within quantum theory
itself.  Bound entanglement does not preclude a secret key, as some
bound-entangled states do yield one~\cite{HorodeckiEtAl2005}.
The reason is that the target of key distillation is
not the singlet but the more general \emph{private
state}~\cite{HorodeckiEtAl2009}.  We now show that, for the standard-basis distributions that
started the subject, a key is distillable throughout the bound-entangled
range.  Gisin and Wolf obtained the
distributions of Table~\ref{tab:gw} by measuring the Horodecki family
$\sigma_\alpha$ of
Ref.~\cite{HorodeckiActivated1999} in the standard bases
(Fig.~\ref{fig:alpha})~\cite{GisinWolf2000}.  Their protocol distills a key
exactly for $\alpha>4$, and
Refs.~\cite{GisinRennerWolf2001,GisinRennerWolf2002} assembled strong
evidence that, on the bound-entangled interval $3<\alpha\le4$, these
distributions have bound information.  This interval became \emph{the}
candidate region for bound information.
\begin{theorem}
\label{thm:gw}
The family of Table~\ref{tab:gw} satisfies $S(X;Y\Vert Z)>0$ for every
$3<\alpha\le5$, whereas $I_{\rm form}(X;Y\mid Z)=0$ for $2\le\alpha\le3$.
\end{theorem}
Thus, contrary to the evidence assembled at the time, no member of the
family has bound information.  A
key is distillable exactly when creating the distribution costs any
secrecy at all (Fig.~\ref{fig:alpha}).
The same holds for the quantum states $\sigma_\alpha$ themselves.
\begin{corollary}
\label{cor:quantum-key}
For every $3<\alpha\le5$, the key distillable from $\sigma_\alpha$ by
local operations and two-way public communication satisfies
$K_D^{\leftrightarrow}(\sigma_\alpha)\ge S(X;Y\Vert Z)>0$.  In
particular, all bound-entangled members, $3<\alpha\le4$, have
positive two-way distillable key.
\end{corollary}
The proof of Theorem~\ref{thm:gw} is given in Appendix~\ref{sec:quantum-candidate} and that of Corollary~\ref{cor:quantum-key}
in Appendix~\ref{sec:gw-quantum-key}.  The cases $2\le\alpha\le3$ and $\alpha>4$ are already
contained in Ref.~\cite{GisinWolf2000}.  For $3<\alpha\le4$, finite-block
filters based on cosets of ternary parity-check codes satisfy GGK's
positive-rate criterion; the two-copy example at $\alpha=4$ and the
general construction are given in Appendices~\ref{sec:two-copy}--\ref{sec:parity-family}.

\emph{Conclusion.---}%
Is there bound information?  The question in the title of
Ref.~\cite{GisinWolf2000} now has an affirmative answer.  The
two-bit--one-trit source of Table~\ref{tab:source} costs secret bits to
create although none can be distilled from it.  The irreversibility that
bound entanglement displays, a positive creation cost with zero
distillable yield, is therefore not an exclusively quantum phenomenon.  It
exists already in classical information theory.

\emph{Physical implications.---}%
Classical and quantum information theory have inspired each other from
the start.  Many classical concepts, such as entropy, channels, and
coding, have quantum counterparts, and conversely, entanglement
manipulation, teleportation, and the no-cloning theorem have found
classical counterparts~\cite{Nielsen1999,JonathanPlenio1999,Spekkens2007,OzolsSmithSmolin2014,ChitambarFortescueHsieh2015,Daffertshofer2002,DelSantoGisin2025}.
Quantum information theory has also been related to other areas of
physics, in particular to thermodynamics, where the bookkeeping of
entanglement parallels that of energy~\cite{HorodeckiThermo1998}.  A
classical counterpart of a quantum phenomenon shows which part of the
phenomenon is genuinely quantum.  For example, nontrivial probability
distributions cannot in general be copied~\cite{Daffertshofer2002}, so
copying already fails classically when knowledge is incomplete.  What is
genuinely quantum is that the prohibition extends to pure states, where
knowledge is maximal.

Bound information is the classical counterpart of bound entanglement,
whose existence in nature was posed as a question at its
discovery~\cite{Horodecki1998}.  As in the case of cloning, the
counterpart separates the quantum part of the phenomenon from the
classical part.  The irreversibility itself, a positive creation cost
with zero yield, exists classically as well and is therefore not the
quantum part.  What is quantum is where the irreversibility resides.
Classically it depends on the specification of the eavesdropper,
because the same correlations between Alice and Bob can be created
without any secrecy against one eavesdropper, while against another
they cost secrecy that can never be recovered.  In our construction,
the first is the fictitious $J$ and the second the actual $Z$.  A
quantum state, by contrast, carries the eavesdropper within.  All
purifications of a bipartite state are equivalent up to an isometry on
Eve's system, so the bipartite state alone fixes the most any adversary can know,
and bound entanglement is a property of that state itself.  This has no
classical counterpart.  For this family, after Alice and Bob's
standard-basis measurements, Eve's full quantum side information is an
orthogonal encoding of $Z$.  The resulting key is therefore secure even
against a purifying adversary (Corollary~\ref{cor:quantum-key}), although
no singlets can be distilled from the bound-entangled members.
Bound entanglement thus provides an instructive example of how a
classical counterpart isolates what is genuinely quantum.

\emph{Outlook.---}%
The proof of the vanishing rate in Theorem~\ref{thm:bound-information}
suggests a general tool.  If the infimum in Eq.~\eqref{eq:intrinsic} is
taken not only over the variables Eve can simulate but over all variables
her channel dominates in the less-noisy sense, one obtains a smaller
upper bound on the secret-key rate, by Eq.~\eqref{eq:GGK}.  This bound
lies between the intrinsic information and the bound of Gohari and
Anantharam~\cite{GohariAnantharam2010}.  On our source it vanishes, while
the intrinsic information stays positive.  The same idea can be applied
to quantum states.  Comparing eavesdroppers by mutual information rather
than by simulation may lead to new bounds on distillable key and
entanglement, as the intrinsic information once led to squashed
entanglement~\cite{ChristandlWinter2004}.

Finally, our source makes the question of activation concrete.  Bound
entanglement can be activated, in the sense that bound-entangled states,
yielding no distillable singlets on their own, help to distill
entanglement from other states.
The Horodecki family of Fig.~\ref{fig:alpha} was introduced to show
this~\cite{HorodeckiActivated1999}.  Whether bipartite bound information can be
activated as well has been asked~\cite{PretticoAcin2013} but could not
be settled without a proven example.  The source of
Table~\ref{tab:source} and the family of Appendix~\ref{sec:family} now provide explicit bound-information
resources on which to investigate activation.

\begin{acknowledgments}
\emph{Acknowledgments.---}%
The construction reported here emerged from an extended
exploratory dialogue with ChatGPT (OpenAI, GPT-5.6 Sol; accessed July 2026).
The authors directed the tool through iterative prompts to test
conjectures, search for counterexamples, and check intermediate
calculations.
Its decisive technical contribution was to identify
Proposition~1 of Gohari, G{\"u}nl{\"u}, and
Kramer~\cite{GohariGunluKramer2020} in response to the authors' targeted
question about an exact zero-key comparison valid for arbitrary
blocklengths and arbitrary interactive public discussion.  This
result from the information-theoretic key-agreement literature supplied
precisely the missing tool and opened the less-noisy but non-degraded
proof strategy used here.
Every AI-assisted argument was independently reconstructed
and verified line by line, with exact symbolic checks where applicable.
The authors take full responsibility for the results.
We warmly thank Stefan Wolf for many long and inspiring discussions on
bound information, reaching back more than twenty-five years to the time when the
conjecture first took shape.  Stefan chose not to join as a coauthor.
Before putting his name to work in which AI tools played a part, he
prefers to take the time to reflect on their rapid development and on how
they should enter scientific practice.  This work was supported by the
Swiss National Science Foundation through project No.~20QU-1\_225171.
\end{acknowledgments}

\paragraph*{Data availability.}
This is a purely mathematical work; no data were created or analyzed in this
study.  The finite-block examples in Appendix~\ref{sec:parity-family} can be reproduced using the accompanying Python
script.

\nocite{HorodeckiPankowski2008}
\bibliography{references}

@article{Maurer1993,
  author = {Ueli M. Maurer},
  title = {Secret Key Agreement by Public Discussion from Common Information},
  journal = {IEEE Trans. Inf. Theory},
  volume = {39},
  number = {3},
  pages = {733--742},
  year = {1993},
  doi = {10.1109/18.256484}
}

@article{AhlswedeCsiszar1993,
  author = {Rudolf Ahlswede and Imre Csisz{\'a}r},
  title = {Common Randomness in Information Theory and Cryptography. Part I: Secret Sharing},
  journal = {IEEE Trans. Inf. Theory},
  volume = {39},
  number = {4},
  pages = {1121--1132},
  year = {1993},
  doi = {10.1109/18.243431}
}

@article{Horodecki1998,
  author = {Micha{\l} Horodecki and Pawe{\l} Horodecki and Ryszard Horodecki},
  title = {Mixed-State Entanglement and Distillation: Is There a ``Bound'' Entanglement in Nature?},
  journal = {Phys. Rev. Lett.},
  volume = {80},
  number = {24},
  pages = {5239--5242},
  year = {1998},
  doi = {10.1103/PhysRevLett.80.5239}
}

@article{MaurerWolf1999,
  author = {Ueli Maurer and Stefan Wolf},
  title = {Unconditionally Secure Key Agreement and the Intrinsic Conditional Information},
  journal = {IEEE Trans. Inf. Theory},
  volume = {45},
  number = {2},
  pages = {499--514},
  year = {1999},
  doi = {10.1109/18.748999}
}

@inproceedings{GisinWolf2000,
  author = {Nicolas Gisin and Stefan Wolf},
  title = {Linking Classical and Quantum Key Agreement: Is There ``Bound Information''?},
  booktitle = {Advances in Cryptology---CRYPTO 2000},
  series = {Lecture Notes in Computer Science},
  volume = {1880},
  pages = {482--500},
  publisher = {Springer},
  year = {2000},
  doi = {10.1007/3-540-44598-6_30},
  eprint = {quant-ph/0005042},
  archivePrefix = {arXiv}
}

@inproceedings{MaurerWolf2000,
  author = {Ueli Maurer and Stefan Wolf},
  title = {Information-Theoretic Key Agreement: From Weak to Strong Secrecy for Free},
  booktitle = {Advances in Cryptology---EUROCRYPT 2000},
  series = {Lecture Notes in Computer Science},
  volume = {1807},
  pages = {351--368},
  publisher = {Springer},
  year = {2000},
  doi = {10.1007/3-540-45539-6_24}
}

@article{GisinRennerWolf2002,
  author = {Nicolas Gisin and Renato Renner and Stefan Wolf},
  title = {Linking Classical and Quantum Key Agreement: Is There a Classical Analog to Bound Entanglement?},
  journal = {Algorithmica},
  volume = {34},
  number = {4},
  pages = {389--412},
  year = {2002},
  doi = {10.1007/s00453-002-0972-7}
}

@incollection{GisinRennerWolf2001,
  author = {Nicolas Gisin and Renato Renner and Stefan Wolf},
  title = {Bound Information: The Classical Analog to Bound Quantum Entanglement},
  booktitle = {European Congress of Mathematics, Barcelona 2000},
  editor = {Carles Casacuberta and Rosa Maria Mir{\'o}-Roig and Joan Verdera and Sebasti{\`a} Xamb{\'o}-Descamps},
  pages = {439--447},
  publisher = {Birkh{\"a}user},
  address = {Basel},
  year = {2001},
  doi = {10.1007/978-3-0348-8266-8_38}
}

@inproceedings{RennerWolf2003,
  author = {Renato Renner and Stefan Wolf},
  title = {New Bounds in Secret-Key Agreement: The Gap between Formation and Secrecy Extraction},
  booktitle = {Advances in Cryptology---EUROCRYPT 2003},
  series = {Lecture Notes in Computer Science},
  volume = {2656},
  pages = {562--577},
  publisher = {Springer},
  year = {2003},
  doi = {10.1007/3-540-39200-9_35}
}

@inproceedings{ChristandlRennerWolf2003,
  author = {Matthias Christandl and Renato Renner and Stefan Wolf},
  title = {A Property of the Intrinsic Mutual Information},
  booktitle = {Proceedings of the 2003 IEEE International Symposium on Information Theory},
  pages = {258},
  publisher = {IEEE},
  year = {2003},
  doi = {10.1109/ISIT.2003.1228272}
}

@article{AcinCiracMasanes2004,
  author = {Antonio Ac{\'i}n and J. Ignacio Cirac and Llu{\'i}s Masanes},
  title = {Multipartite Bound Information Exists and Can Be Activated},
  journal = {Phys. Rev. Lett.},
  volume = {92},
  pages = {107903},
  year = {2004},
  doi = {10.1103/PhysRevLett.92.107903}
}

@article{HorodeckiEtAl2005,
  author = {Karol Horodecki and Micha{\l} Horodecki and Pawe{\l} Horodecki and Jonathan Oppenheim},
  title = {Secure Key from Bound Entanglement},
  journal = {Phys. Rev. Lett.},
  volume = {94},
  pages = {160502},
  year = {2005},
  doi = {10.1103/PhysRevLett.94.160502}
}

@article{Daffertshofer2002,
  author = {A. Daffertshofer and A. R. Plastino and A. Plastino},
  title = {Classical No-Cloning Theorem},
  journal = {Phys. Rev. Lett.},
  volume = {88},
  pages = {210601},
  year = {2002},
  doi = {10.1103/PhysRevLett.88.210601}
}

@article{ChristandlWinter2004,
  author = {Matthias Christandl and Andreas Winter},
  title = {``{S}quashed Entanglement'': An Additive Entanglement Measure},
  journal = {J. Math. Phys.},
  volume = {45},
  pages = {829--840},
  year = {2004},
  doi = {10.1063/1.1643788}
}

@article{HorodeckiEtAl2009,
  author = {Karol Horodecki and Micha{\l} Horodecki and Pawe{\l} Horodecki and Jonathan Oppenheim},
  title = {General Paradigm for Distilling Classical Key From Quantum States},
  journal = {IEEE Trans. Inf. Theory},
  volume = {55},
  pages = {1898--1929},
  year = {2009},
  doi = {10.1109/TIT.2008.2009798}
}

@article{HorodeckiPankowski2008,
  author = {Karol Horodecki and {\L}ukasz Pankowski and Micha{\l} Horodecki and Pawe{\l} Horodecki},
  title = {Low-Dimensional Bound Entanglement with One-Way Distillable Cryptographic Key},
  journal = {IEEE Trans. Inf. Theory},
  volume = {54},
  pages = {2621--2625},
  year = {2008},
  doi = {10.1109/TIT.2008.921709}
}

@phdthesis{RennerThesis2005,
  author = {Renato Renner},
  title = {Security of Quantum Key Distribution},
  school = {ETH Zurich},
  address = {Zurich},
  year = {2005},
  doi = {10.3929/ethz-a-005115027},
  eprint = {quant-ph/0512258},
  archivePrefix = {arXiv}
}

@article{Schroedinger1935,
  author = {Erwin Schr{\"o}dinger},
  title = {Discussion of Probability Relations between Separated Systems},
  journal = {Math. Proc. Cambridge Philos. Soc.},
  volume = {31},
  number = {4},
  pages = {555--563},
  year = {1935},
  doi = {10.1017/S0305004100013554}
}

@article{BennettTeleportation1993,
  author = {Charles H. Bennett and Gilles Brassard and Claude Cr{\'e}peau and Richard Jozsa and Asher Peres and William K. Wootters},
  title = {Teleporting an Unknown Quantum State via Dual Classical and {E}instein-{P}odolsky-{R}osen Channels},
  journal = {Phys. Rev. Lett.},
  volume = {70},
  number = {13},
  pages = {1895--1899},
  year = {1993},
  doi = {10.1103/PhysRevLett.70.1895}
}

@article{Ekert1991,
  author = {Artur K. Ekert},
  title = {Quantum Cryptography Based on {B}ell's Theorem},
  journal = {Phys. Rev. Lett.},
  volume = {67},
  number = {6},
  pages = {661--663},
  year = {1991},
  doi = {10.1103/PhysRevLett.67.661}
}

@inproceedings{BenOrEtAl2005,
  author = {Michael Ben-Or and Micha{\l} Horodecki and Debbie W. Leung and Dominic Mayers and Jonathan Oppenheim},
  title = {The Universal Composable Security of Quantum Key Distribution},
  booktitle = {Theory of Cryptography (TCC 2005)},
  series = {Lecture Notes in Computer Science},
  volume = {3378},
  pages = {386--406},
  publisher = {Springer},
  year = {2005},
  doi = {10.1007/978-3-540-30576-7_21}
}

@article{MasanesAcin2006,
  author = {Llu{\'i}s Masanes and Antonio Ac{\'i}n},
  title = {Multipartite Secret Correlations and Bound Information},
  journal = {IEEE Trans. Inf. Theory},
  volume = {52},
  number = {10},
  pages = {4686--4694},
  year = {2006},
  doi = {10.1109/TIT.2006.881711}
}

@article{MasanesWinter2010,
  author = {Llu{\'i}s Masanes and Andreas Winter},
  title = {A Non-Distillability Criterion for Secret Correlations},
  journal = {Quantum Inf. Comput.},
  volume = {10},
  number = {1--2},
  pages = {152--159},
  year = {2010},
  doi = {10.26421/QIC10.1-2-11}
}

@article{PretticoAcin2013,
  author = {Giuseppe Prettico and Antonio Ac{\'i}n},
  title = {Can Bipartite Classical Information Be Activated?},
  journal = {Quantum Inf. Comput.},
  volume = {13},
  number = {3--4},
  pages = {245--265},
  year = {2013},
  doi = {10.26421/QIC13.3-4-6}
}

@article{GohariGunluKramer2020,
  author = {Amin Gohari and Onur G{\"u}nl{\"u} and Gerhard Kramer},
  title = {Coding for Positive Rate in the Source Model Key Agreement Problem},
  journal = {IEEE Trans. Inf. Theory},
  volume = {66},
  number = {10},
  pages = {6303--6323},
  year = {2020},
  doi = {10.1109/TIT.2020.2990750}
}

@article{Nielsen1999,
  author = {Michael A. Nielsen},
  title = {Conditions for a Class of Entanglement Transformations},
  journal = {Phys. Rev. Lett.},
  volume = {83},
  number = {2},
  pages = {436--439},
  year = {1999},
  doi = {10.1103/PhysRevLett.83.436}
}

@article{JonathanPlenio1999,
  author = {Daniel Jonathan and Martin B. Plenio},
  title = {Entanglement-Assisted Local Manipulation of Pure Quantum States},
  journal = {Phys. Rev. Lett.},
  volume = {83},
  number = {17},
  pages = {3566--3569},
  year = {1999},
  doi = {10.1103/PhysRevLett.83.3566}
}

@article{GisinEtAlQKD2002,
  author = {Nicolas Gisin and Gr{\'e}goire Ribordy and Wolfgang Tittel and Hugo Zbinden},
  title = {Quantum Cryptography},
  journal = {Rev. Mod. Phys.},
  volume = {74},
  number = {1},
  pages = {145--195},
  year = {2002},
  doi = {10.1103/RevModPhys.74.145}
}

@article{Spekkens2007,
  author = {Robert W. Spekkens},
  title = {Evidence for the Epistemic View of Quantum States: A Toy Theory},
  journal = {Phys. Rev. A},
  volume = {75},
  number = {3},
  pages = {032110},
  year = {2007},
  doi = {10.1103/PhysRevA.75.032110}
}

@article{HorodeckiReview2009,
  author = {Ryszard Horodecki and Pawe{\l} Horodecki and Micha{\l} Horodecki and Karol Horodecki},
  title = {Quantum Entanglement},
  journal = {Rev. Mod. Phys.},
  volume = {81},
  number = {2},
  pages = {865--942},
  year = {2009},
  doi = {10.1103/RevModPhys.81.865}
}

@article{PortmannRenner2022,
  author = {Christopher Portmann and Renato Renner},
  title = {Security in Quantum Cryptography},
  journal = {Rev. Mod. Phys.},
  volume = {94},
  number = {2},
  pages = {025008},
  year = {2022},
  doi = {10.1103/RevModPhys.94.025008}
}

@article{DelSantoGisin2025,
  author = {Del Santo, Flavio and Gisin, Nicolas},
  title = {Which Features of Quantum Physics Are Not Fundamentally Quantum but Are Due to Indeterminism?},
  journal = {Quantum},
  volume = {9},
  pages = {1686},
  year = {2025},
  doi = {10.22331/q-2025-04-03-1686}
}

@incollection{KornerMarton1977,
  author = {J{\'a}nos K{\"o}rner and Katalin Marton},
  title = {Comparison of Two Noisy Channels},
  booktitle = {Topics in Information Theory},
  series = {Colloquia Mathematica Societatis J{\'a}nos Bolyai},
  volume = {16},
  pages = {411--423},
  publisher = {North-Holland},
  address = {Amsterdam},
  year = {1977}
}

@article{AhlswedeGacs1976,
  author = {Rudolf Ahlswede and Peter G{\'a}cs},
  title = {Spreading of Sets in Product Spaces and Hypercontraction of the Markov Operator},
  journal = {Ann. Probab.},
  volume = {4},
  number = {6},
  pages = {925--939},
  year = {1976},
  doi = {10.1214/aop/1176995937}
}

@article{GohariAnantharam2010,
  author = {Amin A. Gohari and Venkat Anantharam},
  title = {Information-Theoretic Key Agreement of Multiple Terminals---Part I},
  journal = {IEEE Trans. Inf. Theory},
  volume = {56},
  number = {8},
  pages = {3973--3996},
  year = {2010},
  doi = {10.1109/TIT.2010.2050832}
}

@article{HorodeckiActivated1999,
  author = {Pawe{\l} Horodecki and Micha{\l} Horodecki and Ryszard Horodecki},
  title = {Bound Entanglement Can Be Activated},
  journal = {Phys. Rev. Lett.},
  volume = {82},
  number = {5},
  pages = {1056--1059},
  year = {1999},
  doi = {10.1103/PhysRevLett.82.1056}
}

@article{HorodeckiPPT1997,
  author = {Pawe{\l} Horodecki},
  title = {Separability Criterion and Inseparable Mixed States with Positive Partial Transposition},
  journal = {Phys. Lett. A},
  volume = {232},
  number = {5},
  pages = {333--339},
  year = {1997},
  doi = {10.1016/S0375-9601(97)00416-7}
}

@mastersthesis{RennerDiploma2000,
  author = {Renato Renner},
  title = {Linking Information Theoretic Secret-Key Agreement and Quantum Purification},
  school = {ETH Zurich},
  type = {Diploma thesis},
  address = {Zurich},
  month = aug,
  year = {2000}
}

@article{OzolsSmithSmolin2014,
  author = {Maris Ozols and Graeme Smith and John A. Smolin},
  title = {Bound Entangled States with a Private Key and Their Classical Counterpart},
  journal = {Phys. Rev. Lett.},
  volume = {112},
  pages = {110502},
  year = {2014},
  doi = {10.1103/PhysRevLett.112.110502}
}

@article{ChitambarFortescueHsieh2015,
  author = {Eric Chitambar and Ben Fortescue and Min-Hsiu Hsieh},
  title = {Classical Analog to Entanglement Reversibility},
  journal = {Phys. Rev. Lett.},
  volume = {115},
  pages = {090501},
  year = {2015},
  doi = {10.1103/PhysRevLett.115.090501}
}

@article{CollinsPopescu2002,
  author = {Daniel Collins and Sandu Popescu},
  title = {Classical Analog of Entanglement},
  journal = {Phys. Rev. A},
  volume = {65},
  pages = {032321},
  year = {2002},
  doi = {10.1103/PhysRevA.65.032321}
}

@inproceedings{ChristandlEtAl2007,
  author = {Matthias Christandl and Artur Ekert and Micha{\l} Horodecki and Pawe{\l} Horodecki and Jonathan Oppenheim and Renato Renner},
  title = {Unifying Classical and Quantum Key Distillation},
  booktitle = {Theory of Cryptography (TCC 2007)},
  series = {Lecture Notes in Computer Science},
  volume = {4392},
  pages = {456--478},
  publisher = {Springer},
  year = {2007},
  doi = {10.1007/978-3-540-70936-7_25}
}

@misc{KhesinTungVedula2023,
  author = {Andrey Boris Khesin and Andrew Tung and Karthik Vedula},
  title = {New Properties of Intrinsic Information and Their Relation to Bound Secrecy},
  year = {2023},
  eprint = {2308.09031},
  archivePrefix = {arXiv},
  primaryClass = {cs.IT},
  doi = {10.48550/arXiv.2308.09031}
}

@inproceedings{RennerWolf2002,
  author = {Renato Renner and Stefan Wolf},
  title = {Towards Proving the Existence of ``Bound'' Information},
  booktitle = {Proceedings of the 2002 IEEE International Symposium on Information Theory},
  pages = {103},
  publisher = {IEEE},
  year = {2002},
  doi = {10.1109/ISIT.2002.1023375}
}

@inproceedings{RennerSkripskyWolf2003,
  author = {Renato Renner and Juraj Skripsky and Stefan Wolf},
  title = {A New Measure for Conditional Mutual Information and Its Properties},
  booktitle = {Proceedings of the 2003 IEEE International Symposium on Information Theory},
  pages = {259},
  publisher = {IEEE},
  year = {2003},
  doi = {10.1109/ISIT.2003.1228273}
}

@article{KhatriLutkenhaus2017,
  author = {Sumeet Khatri and Norbert L{\"u}tkenhaus},
  title = {Numerical Evidence for Bound Secrecy from Two-Way Postprocessing in Quantum Key Distribution},
  journal = {Phys. Rev. A},
  volume = {95},
  pages = {042320},
  year = {2017},
  doi = {10.1103/PhysRevA.95.042320}
}

@article{Blackwell1953,
  author = {David Blackwell},
  title = {Equivalent Comparisons of Experiments},
  journal = {Ann. Math. Stat.},
  volume = {24},
  number = {2},
  pages = {265--272},
  year = {1953},
  doi = {10.1214/aoms/1177729032}
}

@article{HorodeckiThermo1998,
  author = {Pawe{\l} Horodecki and Ryszard Horodecki and Micha{\l} Horodecki},
  title = {Entanglement and Thermodynamical Analogies},
  journal = {Acta Phys. Slovaca},
  volume = {48},
  pages = {141--156},
  year = {1998},
  eprint = {quant-ph/9805072},
  archivePrefix = {arXiv}
}

\clearpage
\onecolumngrid
\appendix
\setcounter{secnumdepth}{2}
\setcounter{section}{0}
\renewcommand{\thesection}{\Alph{section}}
\renewcommand{\thesubsection}{\thesection.\arabic{subsection}}
\makeatletter
\renewcommand{\p@subsection}{}
\makeatother
\numberwithin{equation}{section}
\numberwithin{theorem}{section}
\section{Conventions and proof map}
\label{sec:conventions}

We specify the security definitions and outline the proofs.

All alphabets are finite and logarithms are base two.  Alice and Bob observe
$n$ independent copies of $(X,Y)$, have private randomness, and exchange an
arbitrary finite-round authenticated public transcript $F$.  Their outputs
$K_A,K_B$, with values in a set $\mathcal K_n$, must agree with probability tending to one, be asymptotically
uniform in the sense that
\(\log|\mathcal K_n|-H(K_A)\to0\), and obey strong mutual-information secrecy
$I(K_A;Z^n,F)\to0$.  The rate of such a sequence of protocols is
$\liminf_{n\to\infty}\frac1n\log|\mathcal K_n|$, and the supremum of
achievable rates is denoted $S(X;Y\Vert Z)$.  The usual weak and strong secrecy definitions give the
same asymptotic rate for finite sources~\cite{MaurerWolf2000}.  The
comparison proof in Appendix~\ref{sec:capacity-comparison} shows
that replacing $Z^n$ by $J^n$ does not increase the leakage.  It therefore applies under either
normalization.  The
composable, distance-based definition of secrecy~\cite{RennerThesis2005},
which requires the generated key, jointly with Eve's view, to be close in
statistical distance to an ideal key, yields the same rate as well, since
distance-security implies weak mutual-information security by continuity
of entropy, weak and strong mutual-information security coincide by
Ref.~\cite{MaurerWolf2000}, and strong mutual-information security implies
distance-security via Pinsker's inequality.

For a nonnegative matrix $M=(M(x,y))_{x,y}$, a \emph{slice} means an
unnormalized joint law such as
\begin{equation}
 M_z(x,y):=P(X=x,Y=y,Z=z).
\end{equation}
A nonzero slice represents a conditionally independent law exactly when it
has matrix rank one.

The proof of the main result has four logically separate parts:
\begin{enumerate}
 \item an auxiliary variable $J$ makes $X$ and $Y$ conditionally independent;
 \item the actual Eve channel is less noisy than the $J$ channel;
 \item less-noisy dominance implies
 $S(X;Y\Vert Z)\le S(X;Y\Vert J)=0$;
 \item every degradation of $Z$ that made $X,Y$ independent would have to
 simulate $J$, but no channel $Z\to J$ exists.
\end{enumerate}
Appendices~\ref{sec:bsc} and \ref{sec:capacity-comparison} prove the two
information-theoretic ingredients used in parts 2 and 3.  Appendices~\ref{sec:exact-source}--\ref{sec:intrinsic-proof} verify the finite source and
parts 1 and 4.  Appendix~\ref{sec:finite-zero} proves the compactness statement
that turns the absence of a channel making $X$ and $Y$
conditionally independent into strictly positive
intrinsic information.  Appendix~\ref{sec:family} gives a two-parameter family containing
the rational example.  Appendix~\ref{sec:renwol-candidate} explains why the
same zero-rate argument does not settle a distinct historical candidate of
Renner and Wolf.  Appendix~\ref{sec:quantum-candidate} determines the whole
standard-basis family of Gisin and Wolf, obtained from the Horodecki
states, and includes a two-copy example at $\alpha=4$.
Section~\ref{sec:gw-quantum-key} derives positive distillable
key for the underlying entangled states.
Appendix~\ref{sec:horodecki-bi-measurement} exhibits alternative measurements of purifications of the same qutrit states whose outcomes do carry bound information.

We use two results of GGK
\cite{GohariGunluKramer2020}.  Their Proposition~1 is the rate comparison
$Z\succeq_{\rm ln}J\Rightarrow
S(X;Y\Vert Z)\le S(X;Y\Vert J)$.  Together with $X\indep Y\mid J$, it
proves that the new source has a zero secret-key rate.
Appendix~\ref{sec:capacity-comparison} proves the comparison independently.  Their
Theorem~4 is an $n$-letter positive-rate characterization.  We use its
achievability direction in Appendix~\ref{sec:quantum-candidate}.  Finite-block
sets satisfying the GGK inequality prove a positive asymptotic rate.
Failure of the inequality for a particular choice of sets does not imply
a zero rate.

\section{The BSC contraction and the less-noisy order}
\label{sec:bsc}

This appendix proves the two channel inequalities used in the main text,
an exact identity for the erasure channel and a strong data-processing
inequality for the flip channel.  Both are statements about the
less-noisy order, which we first recall.
The binary erasure channel $\BEC_\epsilon$ transmits its input bit with
probability $1-\epsilon$ and outputs the erasure symbol $\bot$ otherwise.
The binary symmetric channel $\BSC_\delta$ flips its input bit with
probability $\delta$.

\begin{definition}[Less-noisy order \cite{KornerMarton1977}]
Let $P_{Z\mid T}$ and $P_{J\mid T}$ be two channels with the same finite
input alphabet.  We write
\begin{equation}
 P_{Z\mid T}\succeq_{\mathrm{ln}}P_{J\mid T}
\end{equation}
if $I(U;Z)\ge I(U;J)$ for every finite $U$, every law $P_{UT}$, and outputs
generated through the respective channels.  Only the two marginal channels
are relevant.  No particular coupling of $Z$ and $J$ is required.
\end{definition}

\begin{lemma}[Erasure identity]
\label{lem:bec}
If $Z=\BEC_\epsilon(L)$ and the erasure event is independent of $(U,L)$,
then
\begin{equation}
 I(U;Z)=(1-\epsilon)I(U;L).
\end{equation}
\end{lemma}

\begin{proof}
The erasure flag is visible in $Z$ and independent of $U$.  Conditional on a
nonerasure $Z$ reveals $L$, while conditional on an erasure it is constant.
The chain rule gives the result.
\end{proof}

\begin{lemma}[BSC strong data-processing inequality
\cite{AhlswedeGacs1976}]
\label{lem:bsc}
Let $L\in\{0,1\}$ and $J=L\oplus N$, where
$N\sim\Bern(\delta)$ is independent of $(U,L)$ and
$0\le\delta\le1/2$.  Then
\begin{equation}
 I(U;J)\le(1-2\delta)^2I(U;L).
 \label{eq:bsc-sdpi}
\end{equation}
\end{lemma}

\begin{proof}
Put $c=1-2\delta$ and $r(q)=\delta+cq$.  With
\begin{equation}
 q_u=P(L=1\mid U=u),\qquad q=\mathbb E[q_U],
\end{equation}
binary entropy gives
\begin{align}
 I(U;L)&=h_2(q)-\mathbb E h_2(q_U),\\
 I(U;J)&=h_2(r(q))-\mathbb E h_2(r(q_U)).
\end{align}
Define $g(q)=h_2(r(q))-c^2h_2(q)$.  For $0<q<1$,
\begin{equation}
 g''(q)=\frac{c^2}{\ln2}
 \left[
 \frac1{q(1-q)}-\frac1{r(q)(1-r(q))}
 \right].
 \label{eq:gsecond}
\end{equation}
Direct expansion shows
\begin{equation}
 r(q)(1-r(q))-q(1-q)
 =\delta(1-\delta)(1-2q)^2\ge0.
\end{equation}
Thus $g$ is convex.  Jensen's inequality
$\mathbb E g(q_U)\ge g(q)$, after rearrangement, is precisely
Eq.~\eqref{eq:bsc-sdpi}.  Boundary cases follow by continuity.
\end{proof}

\section{Secret-key rates under the less-noisy order}
\label{sec:capacity-comparison}

The next proposition is Proposition~1 of Gohari, G{\"u}nl{\"u}, and Kramer
(GGK) \cite{GohariGunluKramer2020}, which they also derive from the
interactive upper bound of Gohari and Anantharam
\cite{GohariAnantharam2010}.  We reproduce a proof to fix both the
orientation of the order and the security convention.

\begin{lemma}[Tensorization]
\label{lem:tensor}
If $P_{Z\mid T}\succeq_{\mathrm{ln}}P_{J\mid T}$, then for every $n$,
\begin{equation}
 P_{Z^n\mid T^n}\succeq_{\mathrm{ln}}P_{J^n\mid T^n}
\end{equation}
for the memoryless product channels.
\end{lemma}

\begin{proof}
Take any joint law $P_{UT^n}$.  Couple the two product channels in any
memoryless way and set
$H_i=(Z^{i-1},J_{i+1}^n)$.  Csisz{\'a}r's sum identity gives
\begin{equation}
 I(U;Z^n)-I(U;J^n)
 =\sum_{i=1}^n\!\left[
 I(U;Z_i\mid H_i)-I(U;J_i\mid H_i)
 \right].
 \label{eq:sumidentity}
\end{equation}
For each $H_i=h$ of positive probability, memorylessness leaves the Markov
chain $U-T_i-(Z_i,J_i)$.  Applying the one-copy less-noisy inequality to the
conditional input law and averaging over $h$ makes every bracket in
Eq.~\eqref{eq:sumidentity} nonnegative.
\end{proof}

\begin{proposition}[GGK comparison of secret-key rates]
\label{prop:ggk}
For a fixed honest-party marginal $P_{XY}$,
\begin{equation}
 P_{Z\mid XY}\succeq_{\mathrm{ln}}P_{J\mid XY}
 \quad\Longrightarrow\quad
 S(X;Y\Vert Z)\le S(X;Y\Vert J).
 \label{eq:ggk-comparison}
\end{equation}
\end{proposition}

\begin{proof}
Consider an arbitrary $n$-copy interactive protocol for the source with Eve
variable $Z^n$.  Run the identical honest-party protocol on the same
$(X^n,Y^n)$ marginal, replacing Eve's side information by $J^n$.  The joint
law of $(K_A,K_B,F)$ is unchanged, so reliability, uniformity, and rate are
unchanged.

It remains to compare leakage.  For every transcript value $f$ of positive
probability, private randomness is independent of the source and
\begin{equation}
 P(k_A,f,z^n,j^n\mid x^n,y^n)
 =P(k_A,f\mid x^n,y^n)\,P(z^n,j^n\mid x^n,y^n).
\end{equation}
Consequently, the protocol and memoryless-source structure imply
\begin{equation}
 K_A-(X^n,Y^n)-(Z^n,J^n)
 \qquad\text{under }P(\,\cdot\mid F=f).
\end{equation}
Conditioning on $F=f$ changes the input law of $(X^n,Y^n)$ but not either
channel conditional on that input.  Lemma~\ref{lem:tensor} therefore gives
\begin{equation}
 I(K_A;J^n\mid F=f)\le I(K_A;Z^n\mid F=f).
\end{equation}
Averaging over $f$ and adding the common term $I(K_A;F)$ yields
\begin{equation}
 I(K_A;J^n,F)\le I(K_A;Z^n,F).
 \label{eq:leakcompare}
\end{equation}
Every protocol secure against $Z^n$ is consequently secure against $J^n$ at
the same rate.  Taking the supremum proves Eq.~\eqref{eq:ggk-comparison}.
\end{proof}

\begin{remark}
In
Eq.~\eqref{eq:ggk-comparison}, $Z\succeq_{\rm ln}J$ means that Eve obtains
\emph{at least} as much mutual information through $Z$ as through $J$ for
every auxiliary variable $U$.  Replacing $Z$ by $J$ therefore
weakens Eve and cannot lower the secret-key rate.
\end{remark}

\section{Exact construction and vanishing secret-key rate}
\label{sec:exact-source}

This appendix constructs the source of Table~\ref{tab:source} and proves the
vanishing-rate half of Theorem~\ref{thm:bound-information}.  The three matrices below are the
slices of that table.

Let $X,Y\in\{0,1\}$ and $Z\in\{0,1,\bot\}$.  Define
\begin{align}
 M_0&=\frac1{36}\begin{pmatrix}5&2\\2&0\end{pmatrix},&
 M_1&=\frac1{36}\begin{pmatrix}0&2\\2&5\end{pmatrix},&
 M_\bot&=\frac1{36}\begin{pmatrix}5&4\\4&5\end{pmatrix}.
 \label{eq:exact-source}
\end{align}
The entries sum to
\begin{equation}
\frac{9+9+18}{36}=1,
\end{equation}
and
\begin{equation}
 P_{XY}=M_0+M_1+M_\bot
 =\frac1{18}\begin{pmatrix}5&4\\4&5\end{pmatrix}.
 \label{eq:exact-pxy}
\end{equation}

Introduce a latent bit $L$ through
\begin{align}
 C_0:=P_{XY,L=0}
 &=\frac1{18}\begin{pmatrix}5&2\\2&0\end{pmatrix},&
 C_1:=P_{XY,L=1}
 &=\frac1{18}\begin{pmatrix}0&2\\2&5\end{pmatrix}.
 \label{eq:exact-latent}
\end{align}
Each has mass $1/2$.  Sending $L$ through a $\BEC(1/2)$ gives
$M_0=C_0/2$, $M_1=C_1/2$, and
$M_\bot=(C_0+C_1)/2$, as required.
The corresponding channel reads
\begin{equation}
 P(L=1\mid x,y)=
 \begin{pmatrix}0&1/2\\[1mm]1/2&1\end{pmatrix}_{x,y}.
 \label{eq:exact-L-channel}
\end{equation}
This channel stays fixed when the less-noisy order is tested under
arbitrary input laws on $(X,Y)$.

Send the same latent bit through an auxiliary $\BSC(1/5)$ and call its
output $J$.  Its slices are
\begin{align}
 A:=P_{XY,J=0}
 &=\frac45C_0+\frac15C_1
 =\frac1{18}\begin{pmatrix}4&2\\2&1\end{pmatrix}
 =\frac1{18}\begin{pmatrix}2\\1\end{pmatrix}
 \begin{pmatrix}2&1\end{pmatrix},
 \label{eq:exact-A}\\
 B:=P_{XY,J=1}
 &=\frac15C_0+\frac45C_1
 =\frac1{18}\begin{pmatrix}1&2\\2&4\end{pmatrix}
 =\frac1{18}\begin{pmatrix}1\\2\end{pmatrix}
 \begin{pmatrix}1&2\end{pmatrix}.
 \label{eq:exact-B}
\end{align}
After normalization, both are product distributions.  Therefore
\begin{equation}
 X\indep Y\mid J,\qquad I(X;Y\mid J)=0.
 \label{eq:exact-separator}
\end{equation}

For any auxiliary $U$ and any input law $P_{U,XY}$, generate $L$ from
$(X,Y)$ according to Eq.~\eqref{eq:exact-L-channel}, then generate $Z,J$.
The outputs depend on $(U,X,Y)$ only through $L$.  Lemmas
\ref{lem:bec} and \ref{lem:bsc} give
\begin{equation}
 I(U;Z)=\frac12I(U;L),\qquad
 I(U;J)\le\frac9{25}I(U;L).
\end{equation}
Hence
\begin{equation}
 P_{Z\mid XY}\succeq_{\rm ln}P_{J\mid XY}.
\end{equation}
Proposition~\ref{prop:ggk} and the intrinsic-information upper bound
\cite{MaurerWolf1999} now imply
\begin{equation}
 0\le S(X;Y\Vert Z)
 \le S(X;Y\Vert J)
 \le I(X;Y{\downarrow}J)
 \le I(X;Y\mid J)=0.
 \label{eq:exact-zero}
\end{equation}
Thus the secret-key rate is zero even with arbitrary
interactive public discussion.

\section{Positive intrinsic information}
\label{sec:intrinsic-proof}

This appendix proves the second half of Theorem~\ref{thm:bound-information}, namely that the
intrinsic information of the source is strictly positive.  In the basis $A,B$ of
Appendix~\ref{sec:exact-source}, the slices read
\begin{equation}
 M_0=\frac23A-\frac16B,\qquad
 M_1=-\frac16A+\frac23B,\qquad
 M_\bot=\frac12A+\frac12B.
 \label{eq:ABbasis}
\end{equation}

\begin{lemma}[Nonnegative rank-one matrices in the span of $A,B$]
\label{lem:rigidity}
For all real $a,b$,
\begin{equation}
 \det(aA+bB)=\frac{ab}{36}.
 \label{eq:detidentity}
\end{equation}
Consequently, if $aA+bB$ is nonzero, entrywise nonnegative, and rank one,
then it is a positive multiple of $A$ or of $B$.
\end{lemma}

\begin{proof}
Using Eqs.~\eqref{eq:exact-A} and \eqref{eq:exact-B},
\begin{equation}
 aA+bB=\frac1{18}
 \begin{pmatrix}
 4a+b&2a+2b\\
 2a+2b&a+4b
 \end{pmatrix}.
\end{equation}
Its determinant is
\begin{equation}
 \frac{(4a+b)(a+4b)-4(a+b)^2}{18^2}
 =\frac{9ab}{18^2}=\frac{ab}{36}.
\end{equation}
Rank one forces $ab=0$.  Since both $A$ and $B$ are entrywise
positive and the matrix is nonzero and nonnegative, the remaining
coefficient is positive.
\end{proof}

\begin{proposition}[Conditional independence would imply a simulation of $J$]
\label{prop:separator-simulates-J}
If a channel $Z\to\bar Z$ obeyed
$X\indep Y\mid\bar Z$, then there would exist a channel $Z\to J$ whose
joint slices are exactly $A,B$.
\end{proposition}

\begin{proof}
For every output $k$ of $\bar Z$, its unnormalized slice is
\begin{equation}
 N_k=P_{XY,\bar Z=k}
 =\sum_{z}P(\bar Z=k\mid Z=z)M_z.
\end{equation}
Equation~\eqref{eq:ABbasis} puts $N_k$ in
$\operatorname{span}\{A,B\}$, so $N_k=a_kA+b_kB$.  Conditional independence
means every nonzero $N_k$ is nonnegative and rank one.  By
Lemma~\ref{lem:rigidity}, it is a positive multiple of exactly one of $A$
and $B$.

Group all $A$-type outputs of $\bar Z$ into a symbol $0$ and all
$B$-type outputs into a symbol $1$.  The resulting binary variable $J'$ is a
degradation of $Z$ and has slices $\lambda A,\mu B$ for some
$\lambda,\mu\ge0$.  Their sum is the fixed marginal
$P_{XY}=A+B$.  Since $A$ and $B$ are linearly independent,
$\lambda=\mu=1$.  Thus $J'$ reproduces $J$.
\end{proof}

\begin{proposition}[$J$ is not a degradation of $Z$]
\label{prop:notdegraded}
No stochastic channel $Z\to J$ reproduces the slices $A,B$.
\end{proposition}

\begin{proof}
Assume such a channel and define
\begin{equation}
 r_0=P(J=1\mid Z=0),\quad
 r_1=P(J=1\mid Z=1),\quad
 r_\bot=P(J=1\mid Z=\bot).
\end{equation}
Conditioned on $XY=00$, the latent bit is $L=0$, so
$P(Z=0\mid00)=P(Z=\bot\mid00)=1/2$.  The auxiliary BSC requires
$P(J=1\mid00)=1/5$, hence
\begin{equation}
 \frac12r_0+\frac12r_\bot=\frac15.
 \label{eq:r00}
\end{equation}
Conditioned on $XY=11$, one has $L=1$,
$P(Z=1\mid11)=P(Z=\bot\mid11)=1/2$, and
$P(J=1\mid11)=4/5$.  Thus
\begin{equation}
 \frac12r_1+\frac12r_\bot=\frac45.
 \label{eq:r11}
\end{equation}
Subtracting Eqs.~\eqref{eq:r00} and \eqref{eq:r11} gives
$r_1-r_0=6/5$, impossible for $r_0,r_1\in[0,1]$.
\end{proof}

By the finite zero criterion proved next, zero intrinsic
information would imply a channel making $X$ and $Y$ conditionally
independent.  Propositions
\ref{prop:separator-simulates-J} and \ref{prop:notdegraded} rule this out, so
\begin{equation}
 I(X;Y{\downarrow}Z)>0.
 \label{eq:strictintrinsic}
\end{equation}
Renner and Wolf proved that information of formation is lower bounded by
intrinsic information \cite{RennerWolf2003}.  Combining
Eqs.~\eqref{eq:exact-zero} and \eqref{eq:strictintrinsic} gives
\begin{equation}
 S(X;Y\Vert Z)=0
 <I(X;Y{\downarrow}Z)
 \le I_{\rm form}(X;Y\mid Z).
\end{equation}

\section{Zero intrinsic information for finite alphabets}
\label{sec:finite-zero}

For finite alphabets, zero intrinsic information implies
that some channel makes $X$ and $Y$ conditionally independent.  We include
a proof of this consequence of the finite-attainment result of
Christandl, Renner, and Wolf~\cite{ChristandlRennerWolf2003}.

\begin{proposition}[Zero intrinsic information and conditional independence]
\label{prop:finitezero}
If $X,Y,Z$ are finite and $I(X;Y{\downarrow}Z)=0$, then there is a finite
channel $Z\to W$ such that $X\indep Y\mid W$.  The converse is immediate.
\end{proposition}

\begin{proof}
Discard zero-probability symbols and write $m=|\mathcal Z|$.  For a
probability vector $q$ on $\mathcal Z$, define
\begin{equation}
 P^{(q)}_{XY}:=\sum_z q(z)P_{XY\mid Z=z},
 \qquad
 \Phi(q):=I\!\left(P^{(q)}_{XY}\right).
\end{equation}
The function $\Phi$, the mutual information between $X$ and $Y$ under
the mixed law $P^{(q)}_{XY}$, is continuous and nonnegative on the
compact simplex $\Delta_m$.

Any degradation $Z\to W$ defines conditional distributions
$q_w=P_{Z\mid W=w}$ and weights $\lambda_w=P_W(w)$ satisfying
\begin{equation}
 \sum_w\lambda_wq_w=P_Z,\qquad
 I(X;Y\mid W)=\sum_w\lambda_w\Phi(q_w).
 \label{eq:posteriorensemble}
\end{equation}
Conversely, any finite weighted collection of distributions
satisfying $\sum_w\lambda_wq_w=P_Z$ defines a channel through
\begin{equation}
 P(W=w\mid Z=z)=\frac{\lambda_wq_w(z)}{P_Z(z)}.
 \label{eq:ensemblechannel}
\end{equation}

If the infimum is zero, choose channels whose conditional
mutual information tends to zero.  Let $\mu_n$ be the probability measure
on $\Delta_m$ assigning weight $\lambda_w$ to each conditional distribution
$q_w$ for the $n$th channel.  Each $\mu_n$ has average $P_Z$, and the average
of $\Phi$ tends to zero.  Probability measures on the compact simplex are
weakly sequentially compact.  Since the identity map and $\Phi$ are continuous on
the simplex, a weakly convergent subsequence has a limit $\mu$ satisfying
\begin{equation}
 \int q\,d\mu(q)=P_Z,\qquad
 \int\Phi(q)\,d\mu(q)=0.
\end{equation}
Since $\Phi\ge0$ and its integral is zero, $\mu$ is supported
on the compact set
$\mathcal Q_0=\{q:\Phi(q)=0\}$.  Therefore $P_Z$ lies in
$\operatorname{conv}(\mathcal Q_0)$.  Since $\Delta_m$ has affine dimension
$m-1$, Carath{\'e}odory's theorem expresses $P_Z$ as a convex combination of
at most $m$ points in $\mathcal Q_0$.  Equation~\eqref{eq:ensemblechannel}
then gives a finite $W$ whose conditional $XY$ laws all have mutual
information zero, i.e., are product laws.
\end{proof}

\section{A two-parameter family of bound-information sources}
\label{sec:family}

This appendix shows that the source of Table~\ref{tab:source} is not
isolated.  It belongs to a two-parameter family with the same
properties.  Fix $0<t<1/2$, and set
\begin{equation}
 d=t^2+(1-t)^2,\qquad
 \delta=\frac{t^2}{d}.
\end{equation}
Choose
\begin{equation}
 2\delta<\epsilon\le4\delta(1-\delta).
 \label{eq:familywindow}
\end{equation}
The interval is nonempty because $\delta<1/2$.  Define latent-bit slices
\begin{align}
 C_0&=\begin{pmatrix}
 d/2&t(1-t)/2\\
 t(1-t)/2&0
 \end{pmatrix},&
 C_1&=\begin{pmatrix}
 0&t(1-t)/2\\
 t(1-t)/2&d/2
 \end{pmatrix}.
 \label{eq:familylatent}
\end{align}
Each has mass $1/2$.  Let $Z=\BEC_\epsilon(L)$ and
$J=\BSC_\delta(L)$.  With
\begin{equation}
 u=\begin{pmatrix}1-t\\t\end{pmatrix},\qquad
 v=\begin{pmatrix}t\\1-t\end{pmatrix},
\end{equation}
the $J$ slices are
\begin{equation}
 A=\frac12uu^{\mathsf T},\qquad
 B=\frac12vv^{\mathsf T}.
\end{equation}
Thus $X\indep Y\mid J$.  The upper inequality in
Eq.~\eqref{eq:familywindow} is equivalent to
\begin{equation}
 1-\epsilon\ge(1-2\delta)^2,
\end{equation}
so Lemmas~\ref{lem:bec} and \ref{lem:bsc} give
$Z\succeq_{\rm ln}J$, and Proposition~\ref{prop:ggk} implies a vanishing
secret-key rate.

The corresponding determinant identity is
\begin{equation}
 \det(aA+bB)=\frac{(1-2t)^2}{4}ab.
 \label{eq:familydet}
\end{equation}
Again, every nonzero, nonnegative rank-one matrix in this
span is a positive multiple of $A$ or $B$.  Finally, if
$r_z=P(J=1\mid Z=z)$, the conditions at $XY=00$ and $XY=11$ imply
\begin{equation}
 r_1-r_0=\frac{1-2\delta}{1-\epsilon}.
\end{equation}
The lower inequality $\epsilon>2\delta$ makes the right-hand side larger
than one, excluding $Z\to J$.  The proof of
Appendices~\ref{sec:intrinsic-proof} and \ref{sec:finite-zero} applies verbatim.
Consequently every parameter pair in Eq.~\eqref{eq:familywindow} has
bipartite bound information.  The main-text source is
$t=1/3$, $\delta=1/5$, and $\epsilon=1/2$.

\section{The Renner--Wolf candidate}
\label{sec:renwol-candidate}

This appendix analyzes the Renner--Wolf candidate discussed in
Appendix~\ref{sec:historical-candidates}.  Renner and Wolf proposed it in the concluding
remarks of Ref.~\cite{RennerWolf2003}.  We present it in the same slice
notation as our new source.  Let $X,Y\in\{0,1,2,3\}$, let
$c_a=(1+8a)^{-1}$, and let $E_{xy}$ denote the matrix unit at row $x$,
column $y$.  The two shared Eve symbols have slices
\begin{equation}
 M_0=c_a
 \begin{pmatrix}
  1/8&0&0&0\\
  0&1/8&0&0\\
  0&0&1/4&0\\
  0&0&0&0
 \end{pmatrix},\qquad
 M_1=c_a
 \begin{pmatrix}
  0&1/8&0&0\\
  1/8&0&0&0\\
  0&0&0&0\\
  0&0&0&1/4
 \end{pmatrix}.
 \label{eq:rw-source}
\end{equation}
For every cross-block cell
\begin{equation}
 (x,y)\in
 \bigl(\{0,1\}\times\{2,3\}\bigr)
 \cup
 \bigl(\{2,3\}\times\{0,1\}\bigr),
\end{equation}
Eve receives the unique symbol $(x,y)$, with slice
\begin{equation}
 M_{(x,y)}=c_a aE_{xy}.
 \label{eq:rw-cross-slices}
\end{equation}
Equivalently, the honest-party marginal is
\begin{equation}
 P_{XY}=c_a
 \begin{pmatrix}
  1/8&1/8&a&a\\
  1/8&1/8&a&a\\
  a&a&1/4&0\\
  a&a&0&1/4
 \end{pmatrix}.
 \label{eq:rw-marginal}
\end{equation}
Ref.~\cite{RennerWolf2003} conjectured a vanishing secret-key rate for
$a\ge1/(4\sqrt2)=\sqrt2/8$, motivated by the bound entanglement of the
corresponding quantum translation (F.~Spedalieri, personal communication
cited there).

We first show that the source has positive intrinsic
information for every $a\ge0$.
Suppose a degradation $Z\to\bar Z$ made every nonzero slice
$N_\lambda=P_{XY,\bar Z=\lambda}$ rank one, and put
$q_{\lambda|z}=P(\lambda\mid Z=z)$.  If
$q_{\lambda|0}>0$, then
$N_\lambda(0,0)$, $N_\lambda(1,1)$, and $N_\lambda(2,2)$ are all positive.
Rank one applied to the $\{0,1\}\times\{0,1\}$ submatrix forces
\begin{equation}
 N_\lambda(0,1)N_\lambda(1,0)
 =N_\lambda(0,0)N_\lambda(1,1)>0.
\end{equation}
Only the symbol $Z=1$ supplies the cells $(0,1)$ and $(1,0)$, so
$q_{\lambda|1}>0$ and hence $N_\lambda(3,3)>0$.  Rank one on rows and
columns $\{2,3\}$ would then require
\begin{equation}
 N_\lambda(2,3)N_\lambda(3,2)
 =N_\lambda(2,2)N_\lambda(3,3)>0,
\end{equation}
which is impossible because the cells $(2,3)$ and $(3,2)$ vanish in every
source slice.  Therefore every rank-one degraded slice has
$q_{\lambda|0}=0$, contradicting
$\sum_\lambda q_{\lambda|0}=1$.  Proposition~\ref{prop:finitezero} gives
\begin{equation}
 I(X;Y{\downarrow}Z)>0\qquad\text{for every }a\ge0.
 \label{eq:rw-positive-intrinsic}
\end{equation}

The following observation shows why the less-noisy argument used for our
new source cannot establish a vanishing rate for this candidate.

\begin{lemma}[Deterministic less-noisy channels]
\label{lem:deterministic-less-noisy}
Let $Z=f(T)$ be deterministic.  If
$P_{Z\mid T}\succeq_{\rm ln}P_{J\mid T}$, that is, if
$I(U;Z)\ge I(U;J)$ for every input law $P_{UT}$,
then $J$ is a stochastic degradation of $Z$.
\end{lemma}

\begin{proof}
Take any $t,t'$ with $f(t)=f(t')$, put equal probability on them,
and set $U=T$.  Then $I(U;Z)=0$, so less-noisy dominance forces
$I(U;J)=0$.  Hence
$P_{J\mid T=t}=P_{J\mid T=t'}$.  Thus the conditional
distribution of $J$ depends on $t$ only through $f(t)$, which defines a channel $Z\to J$.
\end{proof}

The variable $Z$ in Eqs.~\eqref{eq:rw-source}--\eqref{eq:rw-cross-slices}
is deterministic given $(X,Y)$.
Consequently, if a variable $J$ both made $X$ and $Y$ conditionally
independent and lay below $Z$ in the less-noisy order, the lemma would make
$J$ an allowed degradation of $Z$.  This would imply
$I(X;Y{\downarrow}Z)=0$, contradicting the positive intrinsic information
of the candidate.  Thus a zero-rate proof for this source must use a
different method.

For completeness, the achievability criterion settles the entire range
below the proposed transition.  Choose the singleton matching cells
$(x_1,y_1)=(0,0)$ and $(x_2,y_2)=(2,2)$.  Eve sees the same symbol in both
cases, and the one-copy GGK inequality reduces, after cancelling the common
normalization, to
\begin{equation}
 \frac18\frac14>a^2.
\end{equation}
Therefore
\begin{equation}
 S(X;Y\Vert Z)>0
 \qquad\text{for}\qquad
 0\le a<\frac1{4\sqrt2}=\frac{\sqrt2}{8}.
\end{equation}
Failure of this particular choice of sets at and beyond the threshold
does not imply a zero rate.  The region $a\ge1/(4\sqrt2)$ remains unresolved.

\section{The Horodecki candidate family}
\label{sec:quantum-candidate}

This appendix proves Theorem~\ref{thm:gw}.  The object of study is the
family of distributions in Table~\ref{tab:gw}.  Gisin and
Wolf introduced it as Example~3 of Ref.~\cite{GisinWolf2000}, by
measuring the qutrit states $\sigma_\alpha$ of
Ref.~\cite{HorodeckiActivated1999} shown in Fig.~\ref{fig:alpha}.  We
first fix notation and state the source.  We then determine the exact
transition of the intrinsic information at $\alpha=3$.  This settles the
zero-cost half of the theorem.  The remaining subsections settle the
distillation half.  We state the positive-rate criterion of GGK, verify
it on two copies at $\alpha=4$, reduce it exactly to sums over cosets,
and prove that parity-code filters satisfy it for every fixed
$3<\alpha\le5$.

\subsection{The source and Eve's observation}

Relabel the three honest symbols by $\Fthree=\{0,1,2\}$.  For
$2\le\alpha\le5$, put $\beta=5-\alpha$.  Using the same notation as for
the new source, the Eve slices are
\begin{equation}
 M_\bot=\frac{2}{21}I_3,\qquad
 M_{(x,y)}=
 \begin{cases}
  \dfrac{\alpha}{21}E_{xy},&y-x=1,\\[1mm]
  \dfrac{\beta}{21}E_{xy},&y-x=-1,
 \end{cases}
 \quad x\ne y ,
 \label{eq:gwsource}
\end{equation}
where arithmetic is modulo three and $E_{xy}$ is the matrix unit.  Thus, on $x=y$, Eve receives a common symbol.  On $x\ne y$, her symbol
identifies the ordered pair $(x,y)$.  Equivalently,
\begin{equation}
 P_{XY}(x,y)=\frac{w(y-x)}{21},\qquad
 w(0)=2,\quad w(1)=\alpha,\quad w(2)=\beta .
 \label{eq:gw-weight}
\end{equation}
This is Example~3 of Gisin and Wolf \cite{GisinWolf2000} (see also
Refs.~\cite{GisinRennerWolf2002,RennerDiploma2000}), shown as
Table~\ref{tab:gw} and obtained by a particular
measurement of the qutrit-state family introduced in
Ref.~\cite{HorodeckiActivated1999}.  The state is bound entangled for
$3<\alpha\le4$.  The classical protocol exhibited in
Ref.~\cite{GisinWolf2000} succeeds for $\alpha>4$.  With the symmetrizing channels
chosen there, $\operatorname{Prob}[\bar X=\bar Y]-\tfrac12$ is
proportional to $(\alpha-1)(\alpha-4)$.  Conditioned on agreement, the
bit is uniform given Eve's entire view.

\subsection{Exact intrinsic-information transition}

Both sides of this transition were identified in Ref.~\cite{GisinWolf2000}, which
gave the degradation below and stated positivity above $\alpha=3$.  A rank-one proof of the positive half also appeared in
Renner's diploma thesis \cite{RennerDiploma2000}.  We include
the proof for completeness.

For $2\le\alpha\le3$, both $\alpha$ and $\beta$ are at least $2$.
Degrade Eve's observation as follows.  Send the common diagonal symbol
entirely to a new output $\star$.  From each off-diagonal symbol of weight $\alpha$, send the fraction
$2/\alpha$ to $\star$, and from each symbol of weight $\beta$ the
fraction $2/\beta$.  Keep every
remainder as a separate output.  Omitting the common normalization $1/21$,
the $\star$ slice is
\begin{equation}
 2\begin{pmatrix}1&1&1\\1&1&1\\1&1&1\end{pmatrix},
 \label{eq:gw-star}
\end{equation}
which has rank one, while every remaining slice is supported on a single
cell.  Hence $X\indep Y$ conditioned on the degraded observation
$\bar Z$.  This also gives a zero-cost formation protocol in the sense of
Ref.~\cite{RennerWolf2003}.  Sample $\bar Z$ publicly, then let Alice and Bob
sample independently from $P_{X\mid\bar Z}$ and
$P_{Y\mid\bar Z}$.  The public transcript is simulable from the target
$Z$ through the displayed channel $Z\to\bar Z$.  No secret bits are used.
Consequently
\begin{equation}
 I_{\rm form}(X;Y\mid Z)
 =I(X;Y{\downarrow}Z)
 =S(X;Y\Vert Z)=0,
 \qquad 2\le\alpha\le3.
 \label{eq:gw-zero-range}
\end{equation}

Conversely, suppose $3<\alpha\le5$, so $0\le\beta<2$, and assume that a
degradation $Z\to\bar Z$ makes $X$ and $Y$ conditionally independent.  Fix
an output $\lambda$ of the degradation.  Write
\begin{equation}
 q_\lambda=P(\lambda\mid Z=\bot)
\end{equation}
for the fraction of the common diagonal symbol sent to $\lambda$, and let
$t_0^\lambda,t_1^\lambda,t_2^\lambda$ be the corresponding fractions of
the three distinct off-diagonal symbols on the $\beta$-weighted directed
cycle.  The resulting slice has all three diagonal entries equal to
$2q_\lambda$ and cycle entries
$\beta t_0^\lambda,\beta t_1^\lambda,\beta t_2^\lambda$.
Every rank-one $3\times3$ matrix satisfies the directed-cycle identity, so
\begin{equation}
 \beta^3t_0^\lambda t_1^\lambda t_2^\lambda
 =(2q_\lambda)^3.
 \label{eq:gw-cycle-identity}
\end{equation}
If $\beta=0$, this identity forces $q_\lambda=0$ for every $\lambda$,
contrary to $\sum_\lambda q_\lambda=1$.  We may therefore assume
$0<\beta<2$.  The arithmetic--geometric mean inequality then gives
\begin{equation}
 t_0^\lambda+t_1^\lambda+t_2^\lambda
 \ge\frac{6q_\lambda}{\beta}.
\end{equation}
Summing over $\lambda$ and using
$\sum_\lambda q_\lambda=1$ and
$\sum_\lambda t_i^\lambda=1$ yields
\begin{equation}
 3\ge\frac6\beta,
\end{equation}
contrary to $\beta<2$.  Proposition~\ref{prop:finitezero} therefore implies
\begin{equation}
 I(X;Y{\downarrow}Z)>0,\qquad 3<\alpha\le5.
 \label{eq:gw-positive-intrinsic}
\end{equation}

\subsection{The positive-rate criterion}

We use the following sufficient condition from GGK
Theorem~4.  If, for some
blocklength $n$, there are disjoint nonempty
${\cal A}_1,{\cal A}_2\subseteq{\cal X}^n$ and disjoint nonempty
${\cal B}_1,{\cal B}_2\subseteq{\cal Y}^n$ such that
\begin{equation}
 \sum_{z^n}
 \sqrt{P(z^n,{\cal A}_1,{\cal B}_1)\,P(z^n,{\cal A}_2,{\cal B}_2)}
 >\sqrt{P({\cal A}_1,{\cal B}_2)\,P({\cal A}_2,{\cal B}_1)},
 \label{eq:ggkpositive}
\end{equation}
then $S(X;Y\Vert Z)>0$
\cite{GohariGunluKramer2020}.  It therefore suffices to find a
finite blocklength and four sets satisfying this inequality.  GGK in fact prove an $n$-letter
necessary-and-sufficient characterization, but only the displayed achievability implication is needed here.

For normalized probability distributions $P$ and $Q$, the
Bhattacharyya coefficient $B(P,Q)=\sum_z\sqrt{P(z)Q(z)}$ measures their
statistical overlap.  The left-hand side of Eq.~\eqref{eq:ggkpositive} applies the
same expression to the unnormalized measures associated with the two
matching events.

\subsection{The two-copy coset filter}
\label{sec:two-copy}

In $\Fthree^2$, define
\begin{equation}
 {\cal L}=\{00,12,21\},\qquad e=(0,1),
\end{equation}
and choose
\begin{equation}
 {\cal A}_1={\cal L},\quad
 {\cal A}_2={\cal B}_1={\cal L}+e,\quad
 {\cal B}_2={\cal L}+2e.
 \label{eq:twocopysets}
\end{equation}
The honest parties output binary labels on the two selected cosets and
reject all other blocks.

For a pair of words $x,y$, let $d=y-x$ coordinatewise.  Its unnormalized
weight is the product of the corresponding $w(d_i)$.  All terms in
Eq.~\eqref{eq:ggkpositive} carry the same factor $21^{-2}$, which we
omit.
For each pair of cosets, each possible difference occurs exactly
three times among the pairs of words.

For the first error event
${\cal A}_1\times{\cal B}_2$, the differences are
\begin{equation}
 {\cal L}+2e=\{02,11,20\},
\end{equation}
with weights $2\beta,\alpha^2,2\beta$.  Hence
\begin{equation}
 \widetilde P({\cal A}_1,{\cal B}_2)
 =3(\alpha^2+4\beta).
 \label{eq:cross12}
\end{equation}
For ${\cal A}_2\times{\cal B}_1$, the differences form
${\cal L}=\{00,12,21\}$, with weights
$4,\alpha\beta,\alpha\beta$, so
\begin{equation}
 \widetilde P({\cal A}_2,{\cal B}_1)
 =3(4+2\alpha\beta).
 \label{eq:cross21}
\end{equation}

Both matching events have differences
\begin{equation}
 {\cal L}+e=\{01,10,22\}.
\end{equation}
For differences $01$ and $10$, one coordinate is diagonal and hidden from
Eve while the other has weight $\alpha$.  The two matching
events have six common Eve outcome sequences, each with unnormalized probability
$2\alpha$ on both sides.  Their total Bhattacharyya overlap is
\begin{equation}
 \sum_{z^2}\sqrt{\widetilde P(z^2,{\cal A}_1,{\cal B}_1)
                       \widetilde P(z^2,{\cal A}_2,{\cal B}_2)}
 =12\alpha.
 \label{eq:matchingoverlap}
\end{equation}
For difference $22$, both coordinates are off diagonal, so Eve learns the
full ordered word pair.  Eve's possible outcome sequences
for the two matching events are disjoint, so their overlap is zero.

Equations~\eqref{eq:cross12}--\eqref{eq:matchingoverlap} reduce the strict
GGK inequality to
\begin{equation}
 8\alpha^2>
 (\alpha^2+4\beta)(2+\alpha\beta),
 \qquad \beta=5-\alpha.
 \label{eq:twocopycondition}
\end{equation}
Equivalently,
\begin{equation}
 q(\alpha):=
 \alpha^4-9\alpha^3+46\alpha^2-92\alpha-40>0.
\end{equation}
The quadratic $q''(\alpha)=12\alpha^2-54\alpha+92$
has negative discriminant, so $q''>0$ for all $\alpha$.  Together
with $q'(3)=49>0$, this shows that $q$ is strictly increasing
on $[3,5]$.  Since
$q(3)=-64$ and $q(4)=8$, it has a unique root
\begin{equation}
 \alpha_2=3.917737700\ldots
\end{equation}
on $[3,5]$, so $q>0$ on $(\alpha_2,5]$, and the two-copy filter proves a positive secret-key rate
for $\alpha_2<\alpha\le5$.  At $\alpha=4$,
\begin{equation}
 12\alpha=48,\qquad
 \sqrt{3(16+4)\,3(4+8)}
 =\sqrt{60\cdot36}=12\sqrt{15}<48.
\end{equation}
Thus the distribution at $\alpha=4$ is not bound information, despite the
bound entanglement of the quantum state from which it was obtained.  This
two-copy filter also refutes the binarized-intrinsic-information
multiplicativity conjecture proposed for this family, with this
particular measurement, in
Ref.~\cite{RennerDiploma2000}.

\subsection{The positive-rate criterion for coset filters}

We now generalize the two-copy construction.  For
$d=(d_1,\ldots,d_n)\in\Fthree^n$, define its unnormalized weight
\begin{equation}
 W(d)=\prod_{i=1}^n w(d_i).
 \label{eq:gw-word-weight}
\end{equation}
Let $L\le\Fthree^n$ be a linear subspace, let $e\notin L$, put
$C=L+e$ (the translate of $L$ by $e$), and choose
\begin{equation}
 {\cal A}_1=L,\qquad
 {\cal A}_2={\cal B}_1=C,\qquad
 {\cal B}_2=L+2e=-C.
 \label{eq:general-cosets}
\end{equation}
Define
\begin{equation}
 G(C)=\{d\in C:\ \exists h\in C
 \text{ such that }d_i h_i=0\text{ for every }i\}.
 \label{eq:good-coset}
\end{equation}
As shown below, $G(C)$ consists of the differences that can
contribute to the overlap between the two matching events.

\begin{lemma}[GGK criterion for coset filters]
\label{lem:coset-ggk}
For the four sets in Eq.~\eqref{eq:general-cosets}, the strict GGK
inequality is equivalent to
\begin{equation}
 \sum_{d\in G(C)}W(d)>
 \sqrt{\left(\sum_{d\in-C}W(d)\right)
       \left(\sum_{d\in L}W(d)\right)}.
 \label{eq:coset-ggk}
\end{equation}
\end{lemma}

\begin{proof}
For each pair of cosets, each possible difference occurs exactly
$|L|$ times among the pairs of words.  Thus the probabilities of the two
error events, before the common factor $21^{-n}$,
are $|L|\sum_{-C}W$ and $|L|\sum_LW$.

Both matching events have difference coset $C$.  Fix $d\in C$.
Eve's word determines $d$ coordinatewise, since an off-diagonal symbol reveals
both entries, while the common diagonal symbol implies $d_i=0$.
Consequently, distinct differences give disjoint sets of
possible Eve outcomes.
Their pairs can be written as $(x,x+d)$ with $x\in L$ and
$(x',x'+d)$ with $x'\in C$.  Eve sees the same word from the two pairs
exactly when $x'_i=x_i$ on every off-diagonal coordinate $d_i\ne0$.
With $h=x'-x$, this is precisely
\begin{equation}
 h\in C,\qquad d_i h_i=0\quad\forall i.
\end{equation}
Hence the overlap vanishes unless $d\in G(C)$.

For $d\in G(C)$, let
\begin{equation}
 K_d=\{\ell\in L:\ell_i=0\text{ whenever }d_i\ne0\}.
\end{equation}
For each common Eve word there are $|K_d|$ compatible pairs in each
matching event, and there are $|L|/|K_d|$ such Eve words.  Their total
Bhattacharyya contribution is therefore
\begin{equation}
 \frac{|L|}{|K_d|}
 \sqrt{|K_d|W(d)\,|K_d|W(d)}
 =|L|W(d).
\end{equation}
Summing over $d\in G(C)$ and cancelling the common factor $|L|/21^n$ gives
Eq.~\eqref{eq:coset-ggk}.
\end{proof}

\subsection{Positive key rate for every
\texorpdfstring{$3<\alpha\le5$}{3<alpha<=5}
}
\label{sec:parity-family}

Choose $m\ge2$ and define the single-parity-check code
\begin{equation}
 L_m=\left\{u\in\Fthree^m:\sum_{i=1}^m u_i=0\right\},
 \qquad
 C_j=\left\{u\in\Fthree^m:\sum_{i=1}^m u_i=j\right\}.
 \label{eq:parity-cosets}
\end{equation}
We call $\sum_{i=1}^m u_i\bmod 3$ the residue of the word
$u$, so that $C_j$ consists of the words of residue $j$ and $L_m=C_0$.
Repeat each coordinate $r$ times through
\begin{equation}
 E_r(u_1,\ldots,u_m)
 =(\underbrace{u_1,\ldots,u_1}_{r},
 \ldots,\underbrace{u_m,\ldots,u_m}_{r})
 \in\Fthree^{mr}.
 \label{eq:repetition-map}
\end{equation}
For $c\in\{1,2\}$, use
\begin{equation}
 {\cal A}_1=E_r(L_m),\quad
 {\cal A}_2={\cal B}_1=E_r(C_c),\quad
 {\cal B}_2=E_r(C_{-c}).
 \label{eq:parity-filter}
\end{equation}

A word $d\in C_c$ satisfies the condition in
Eq.~\eqref{eq:good-coset} exactly when it has a zero coordinate.  Indeed, if
$d_j=0$, the word supported only at $j$ with value $c$ lies in $C_c$ and
has disjoint support from $d$.  If $d$ has full support, the only
word with support disjoint from $d$ is the zero word, which is not in $C_c$.

Put $A_r=2^r$, $B_r=\alpha^r$, and $C_r=\beta^r$, and define the weighted
residue sums
\begin{align}
 S_j^{(m)}
 &=\sum_{\substack{u\in\Fthree^m\\\sum_i u_i=j}}
 A_r^{N_0(u)}B_r^{N_1(u)}C_r^{N_2(u)},\\
 T_j^{(m)}
 &=\sum_{\substack{u\in\{1,2\}^m\\\sum_i u_i=j}}
 B_r^{N_1(u)}C_r^{N_2(u)}.
 \label{eq:residue-sums}
\end{align}
Here $N_k(u)$ counts the coordinates of $u$ equal to $k$.  Lemma
\ref{lem:coset-ggk} now gives the criterion
\begin{equation}
 \left(S_c^{(m)}-T_c^{(m)}\right)^2
 >S_{-c}^{(m)}S_0^{(m)}.
 \label{eq:parity-criterion}
\end{equation}

\begin{theorem}[Distillability of the standard-basis family]
\label{thm:gw-all}
For every fixed $3<\alpha\le5$, there are finite $m,r$ for which the
parity-coset filter in Eq.~\eqref{eq:parity-filter} satisfies
Eq.~\eqref{eq:parity-criterion}.  Consequently the source in
Eq.~\eqref{eq:gwsource} has a positive secret-key rate for every
$3<\alpha\le5$.
\end{theorem}

\begin{proof}
The original protocol of Ref.~\cite{GisinWolf2000} covers $\alpha>4$.
We nevertheless treat the full range $3<\alpha\le5$ at once, so that the
theorem is self-contained.  For this proof, divide every
$S_j^{(m)}$ and $T_j^{(m)}$ by the common factor $B_r^m$, leaving
Eq.~\eqref{eq:parity-criterion} unchanged.  We use the same symbols for
these normalized sums and write
\begin{equation}
 x=\left(\frac2\alpha\right)^r,\qquad
 y=\left(\frac\beta\alpha\right)^r.
 \label{eq:xy-small}
\end{equation}
Take either
\begin{equation}
 m\equiv0\pmod3,\ c=2,
 \qquad\text{or}\qquad
 m\equiv2\pmod3,\ c=1.
 \label{eq:mc-choice}
\end{equation}
If $s=m\bmod3$ is the residue of the all-one word, then in both cases
$c=s-1$.

For $3<\alpha<4$, choose a constant $\kappa$ satisfying
\begin{equation}
 \frac1{\ln(\alpha/2)}
 <\kappa<
 \frac1{\ln(\alpha\beta/4)}.
 \label{eq:kappa-window}
\end{equation}
This interval is nonempty because $\beta<2$.  For $\alpha\ge4$ one has
$\alpha\beta\le4$, so only the left-hand inequality is needed.  Set
\begin{equation}
 r=\lceil\kappa\ln m\rceil.
\end{equation}
Along either residue class in Eq.~\eqref{eq:mc-choice},
\begin{equation}
 mx\longrightarrow0,\qquad
 \frac{y}{mx^2}
 =\frac1m\left(\frac{\alpha\beta}{4}\right)^r
 \longrightarrow0,\qquad
 \frac yx=\left(\frac\beta2\right)^r\longrightarrow0.
 \label{eq:three-limits}
\end{equation}
The same limits hold for $\alpha\ge4$, where the middle expression is at
most $1/m$.

Count the coordinates that differ from $1$.  The total
normalized weight of words with exactly $k$ such coordinates is
$\binom mk(x+y)^k$.  Since $m(x+y)\to0$, words with two or more
such coordinates have total weight $O((m(x+y))^2)$, and
those with three or more have weight $O((m(x+y))^3)$.
The leading contribution to $S_c^{(m)}-T_c^{(m)}$, with
$c=s-1$, comes from the $m$ words with exactly one zero.  Therefore
\begin{equation}
 S_c^{(m)}-T_c^{(m)}=mx\,(1+o(1)).
 \label{eq:good-asymptotic}
\end{equation}
The words with two twos also have residue $c$, but they have full support
and cancel in $S_c^{(m)}-T_c^{(m)}$.
The all-one word dominates its own residue class:
\begin{equation}
 S_s^{(m)}=1+o(1).
 \label{eq:all-one-asymptotic}
\end{equation}
In residue $s+1$, the only contributions with at most two
coordinates different from $1$ come from words with one coordinate equal
to $2$, of total weight $my$, and words with two coordinates
equal to $0$, of total weight
$\binom m2x^2$.  A mixed zero--two pair has residue $s$.
Equation~\eqref{eq:three-limits} makes the former
negligible, while the contribution from words with at least
three coordinates different from $1$ is
$o(m^2x^2)$.  Hence
\begin{equation}
 S_{s+1}^{(m)}=\binom m2x^2(1+o(1)).
 \label{eq:two-zero-asymptotic}
\end{equation}
For both choices in Eq.~\eqref{eq:mc-choice}, the residues of the two error events
$\{-c,0\}$ are exactly $\{s,s+1\}$.  It follows that

\begin{equation}
 \frac{\left(S_c^{(m)}-T_c^{(m)}\right)^2}
 {S_{-c}^{(m)}S_0^{(m)}}
 =\frac{m^2}{\binom m2}\,(1+o(1))
 \longrightarrow 2.
 \label{eq:ratio-to-two}
\end{equation}

The limit is strictly larger than one, so
Eq.~\eqref{eq:parity-criterion} holds for all sufficiently large $m$.
The resulting blocklength $n=mr$ is finite for each fixed
$3<\alpha\le5$.
\end{proof}

Returning to the unnormalized sums of
Eq.~\eqref{eq:residue-sums}, we can evaluate them using the recurrences
\begin{align}
 S_j^{(m+1)}
 &=A_rS_j^{(m)}+B_rS_{j-1}^{(m)}+C_rS_{j-2}^{(m)},\\
 T_j^{(m+1)}
 &=B_rT_{j-1}^{(m)}+C_rT_{j-2}^{(m)},
 \label{eq:residue-recurrence}
\end{align}
starting from $S_0^{(0)}=T_0^{(0)}=1$ and
$S_1^{(0)}=S_2^{(0)}=T_1^{(0)}=T_2^{(0)}=0$.  For rational $\alpha$, the
criterion can be checked exactly by clearing denominators and using these
recurrences.
The blocklength in this construction is finite for each
$3<\alpha\le5$.  The blocklength provided by the proof grows without bound
as $\alpha\downarrow3$, since the window in Eq.~\eqref{eq:kappa-window}
closes in this limit.  The argument does not give a practical lower bound
on the secret-key rate.
For example, the criterion holds for $\alpha=18/5$ with
$(m,r,n)=(5,3,15)$ and $c=1$, for
$\alpha=7/2$ with $(6,4,24)$ and $c=2$,
and for $\alpha=16/5$ with $(119,10,1190)$ and
$c=1$.  In each case $c$ is the value prescribed by
Eq.~\eqref{eq:mc-choice}, and the criterion fails for the other value of
$c$.  These blocklengths are illustrative certificates and need not be
minimal.  The two-copy filter of Sec.~\ref{sec:two-copy} is itself the
member $(m,r,n)=(2,1,2)$ of this family with $c=1$, since $L_2={\cal L}$
and $C_1={\cal L}+e$.  By the monotonicity of $q$ shown there, it
certifies a positive secret-key rate for all $\alpha_2<\alpha\le5$.  For
$\alpha\ge4$, the choices $(5,2,10)$ at $\alpha=4$ and $(5,1,5)$ at
$\alpha=9/2$ and $\alpha=5$, again with $c=1$, also satisfy the
criterion, but $(2,1,2)$ is shorter.
Equations~\eqref{eq:gw-zero-range}, \eqref{eq:gw-positive-intrinsic}, and
Theorem~\ref{thm:gw-all} give the exact transition
\begin{equation}
\begin{array}{c@{\qquad}c@{\qquad}c@{\qquad}c}
 \text{range}&I_{\rm form}(X;Y\mid Z)&I(X;Y{\downarrow}Z)&S(X;Y\Vert Z)\\
 2\le\alpha\le3&0&0&0\\
 3<\alpha\le5&>0&>0&>0.
\end{array}
\label{eq:gw-transition}
\end{equation}
None of the distributions in this measured family has
bound information.

\subsection{Positive distillable key of the qutrit states}
\label{sec:gw-quantum-key}

The same protocols also yield secret key from the quantum states
themselves.  This subsection proves Corollary~\ref{cor:quantum-key}: for every $3<\alpha\le5$,
\begin{equation}
 K_D^{\leftrightarrow}(\sigma_\alpha)\ \ge\ S(X;Y\Vert Z)\ >\ 0,
 \label{eq:quantum-lift}
\end{equation}
where $K_D^{\leftrightarrow}(\sigma_\alpha)$ denotes the rate of secret
key distillable from many copies of $\sigma_\alpha$ by local operations
and two-way public communication, with Eve holding a purifying
system (equivalently, the optimal rate at which \emph{private states}
can be distilled by such
operations~\cite{HorodeckiEtAl2005,HorodeckiEtAl2009}), and
$S(X;Y\Vert Z)$ is the secret-key
rate of the standard-basis source of Table~\ref{tab:gw}.  In
particular, every bound-entangled member, $3<\alpha\le4$, has
positive two-way distillable key.

Choose the purification
\begin{align}
 \lvert\Psi_\alpha\rangle={}&
 \sqrt{\tfrac27}\,\lvert\Phi\rangle\lvert\bot\rangle_E
 +\sqrt{\tfrac{\alpha}{21}}\sum_{i=0}^2
   \lvert i,i{+}1\rangle\lvert i,i{+}1\rangle_E
 \notag\\
 &+\sqrt{\tfrac{5-\alpha}{21}}\sum_{i=0}^2
   \lvert i{+}1,i\rangle\lvert i{+}1,i\rangle_E ,
 \label{eq:gw-purification}
\end{align}
with indices modulo three and the seven displayed vectors on $E$
orthonormal.  This is the canonical purification also used in
Appendix~\ref{sec:horodecki-bi-measurement}.  Since every purification
differs from this one only by an isometry on Eve's system, which changes
no secrecy quantity, no generality is lost.

After Alice and Bob measure in the standard bases, Eve is left in
the same state $\lvert\bot\rangle_E$ for all three diagonal outcomes,
whereas each off-diagonal outcome $(x,y)$ leaves the distinct orthogonal
flag $\lvert x,y\rangle_E$.  The outcome probabilities reproduce
Eq.~\eqref{eq:gw-weight}, and the post-measurement state is an
orthogonal encoding of the classical source of Table~\ref{tab:gw}:
\begin{equation}
 \rho_{XYE}
 =\sum_{x,y}P_{XY}(x,y)\,
 \lvert x,y\rangle\!\langle x,y\rvert\otimes
 \lvert e_{z(x,y)}\rangle\!\langle e_{z(x,y)}\rvert ,
 \label{eq:orthogonal-encoding}
\end{equation}
where $z(x,y)$ is Eve's symbol in Table~II and the states
$\{\lvert e_z\rangle\}_z$ are orthonormal.

Eve's quantum register is therefore operationally equivalent to the
classical variable $Z$.  Measuring $E$ in the flag basis produces $Z$
without loss, while the preparation map
$\lvert z\rangle\mapsto\lvert e_z\rangle$ reconstructs $E$ from $Z$.
The equivalence persists for $n$ copies and after any public
transcript.  Even a transcript-dependent joint POVM $\{M_m\}$ on $E^n$
yields outcome probabilities
$\Pr(m\mid z^n)=\langle e_{z^n}\rvert M_m\lvert e_{z^n}\rangle$ and is
therefore a stochastic post-processing of $Z^n$.  Equivalently, the
final state of any protocol is the image of the corresponding classical
one under the isometry $\lvert z^n\rangle\mapsto\lvert e_{z^n}\rangle$,
so the classical and quantum trace-distance secrecy errors coincide.
Every classical protocol secure against $Z^n$ is thus secure against
the full purifying system $E^n$.  Running a protocol that achieves the
classical rate after the local standard-basis measurements proves
Eq.~\eqref{eq:quantum-lift}.  Positivity of the classical rate for
$3<\alpha\le5$ is Theorem~\ref{thm:gw-all}.

The finite-block classical protocols established in this section are run after local standard-basis measurements.  The previous explicit
low-dimensional bound-entangled states with proven positive distillable key
are approximate private states with a key--shield structure in
dimension $4\otimes4$~\cite{HorodeckiPankowski2008}.  We are not aware of an
earlier such proof for the original $3\otimes3$ family of
Ref.~\cite{HorodeckiActivated1999}.  This argument applies to the
standard-basis source, for which Eve's conditional states form an
orthogonal encoding of $Z$.  It does not apply to the alternative Eve
measurement of Appendix~\ref{sec:horodecki-bi-measurement}, for which
bound information is proved only against the specified
classical adversary.

\section{Bound information from the Horodecki qutrit family}
\label{sec:horodecki-bi-measurement}

The conclusion above concerns the standard-basis measurement used in
Ref.~\cite{GisinWolf2000}.  We now exhibit alternative measurements of the
same qutrit states whose outcomes do have bound information.

Recall the Horodecki family
\begin{equation}
 \sigma_\alpha
 =\frac{2}{7}\lvert\Phi\rangle\!\langle\Phi\rvert
 +\frac{\alpha}{21}
   \sum_{(i,j)\in\{(0,1),(1,2),(2,0)\}}
   \lvert ij\rangle\!\langle ij\rvert
 +\frac{5-\alpha}{21}
   \sum_{(i,j)\in\{(1,0),(2,1),(0,2)\}}
   \lvert ij\rangle\!\langle ij\rvert ,
 \label{eq:alternative-state}
\end{equation}
where
\(\lvert\Phi\rangle=(\lvert00\rangle+\lvert11\rangle+
\lvert22\rangle)/\sqrt3\).
Alice and Bob use the binary POVMs with zero-outcome effects
\begin{equation}
 A_0=\operatorname{diag}\left(0,\frac12,1\right),\qquad
 B_0=\operatorname{diag}\left(1,\frac12,0\right),
 \label{eq:alternative-honest-povms}
\end{equation}
and complementary effects \(A_1=\mathbf1-A_0\) and
\(B_1=\mathbf1-B_0\).

Consider the canonical purification in
Eq.~\eqref{eq:gw-purification}, whose seven orthogonal flag states on
Eve's system are ordered as
\[
 k=\Phi,\;01,\;12,\;20,\;10,\;21,\;02,
\]
and put \(\beta=5-\alpha\).  Eve uses the three-outcome POVM, the
$\alpha$-adapted POVM of Fig.~\ref{fig:alpha},
\begin{equation}
 E_z=\sum_k q_\alpha(z\mid k)\lvert k\rangle\!\langle k\rvert,
 \qquad z\in\{0,1,\bot\},
\end{equation}
where
\begin{equation}
\bigl(q_\alpha(z\mid k)\bigr)_{z,k}
=
\begin{pmatrix}
0&0&0&\dfrac1{10\alpha}&\dfrac1\beta&\dfrac1\beta&0\\[2mm]
0&\dfrac1\alpha&\dfrac1\alpha&0&0&0&\dfrac1{10\beta}\\[2mm]
1&1-\dfrac1\alpha&1-\dfrac1\alpha&1-\dfrac1{10\alpha}&
1-\dfrac1\beta&1-\dfrac1\beta&1-\dfrac1{10\beta}
\end{pmatrix}.
\label{eq:alternative-eve-povm}
\end{equation}
For \(2\le\alpha\le4\), all entries in
Eq.~\eqref{eq:alternative-eve-povm} are nonnegative and each column sums
to one, so this is a valid POVM.

This POVM admits a rank-one refinement whose additional
outcome carries no information about $X$ and $Y$ once $z$ is known.  Let
\(\mathcal S_z=\{k:q_\alpha(z\mid k)>0\}\),
\(r_z=|\mathcal S_z|\), and assign consecutive integers
\(n_z(k)=0,\ldots,r_z-1\) to the elements of \(\mathcal S_z\).
With \(\omega_z=e^{2\pi i/r_z}\), define, for
\(\ell=0,\ldots,r_z-1\),
\begin{equation}
 \lvert\eta_{z,\ell}\rangle
 =\sum_{k\in\mathcal S_z}
 \sqrt{\frac{q_\alpha(z\mid k)}{r_z}}\,
 \omega_z^{\ell n_z(k)}\lvert k\rangle .
 \label{eq:alternative-rank-one}
\end{equation}
Fourier orthogonality gives
\begin{equation}
 \sum_{\ell=0}^{r_z-1}
 \lvert\eta_{z,\ell}\rangle\!\langle\eta_{z,\ell}\rvert=E_z.
\end{equation}
Moreover, the matrix elements between the distinct $AB$ branch
states paired with these flags vanish under
Alice's and Bob's diagonal measurement effects
\(A_x\otimes B_y\).  Consequently,
\begin{equation}
 P(x,y,z,\ell)=\frac{1}{r_z}P(x,y,z).
 \label{eq:alternative-refined-law}
\end{equation}
Given $z$, $\ell$ is uniform and independent of $X$ and $Y$.
Thus $(z,\ell)$ contains exactly the same information about $X$ and $Y$ as
$z$, that is, the two observations are Blackwell equivalent~\cite{Blackwell1953}.
Thus the construction does not rely
on Eve first learning the flag state and then forgetting it.

Direct evaluation of
\[
 P(x,y,z)=
 \operatorname{Tr}\!\left[
 (A_x\otimes B_y\otimes E_z)
 \lvert\Psi_\alpha\rangle\!\langle\Psi_\alpha\rvert
 \right],
\]
where \(\lvert\Psi_\alpha\rangle\) is the canonical purification of
\(\sigma_\alpha\), gives the following unnormalized \(XY\) slices,
independently of \(\alpha\):
\begin{align}
 M_0&=
 \begin{pmatrix}
 11/210&1/42\\
 1/42&0
 \end{pmatrix},&
 M_1&=
 \begin{pmatrix}
 0&1/42\\
 1/42&11/210
 \end{pmatrix},\label{eq:alternative-slices01}\\[2mm]
 M_\bot&=
 \begin{pmatrix}
 22/105&4/21\\
 4/21&22/105
 \end{pmatrix}.
 \label{eq:alternative-slicebot}
\end{align}
This is exactly the member of the two-parameter family in
Appendix~\ref{sec:family} with
\begin{equation}
 t=\frac12-\frac{\sqrt{21}}{42},\qquad
 \delta=\frac12-\frac{\sqrt{21}}{22},\qquad
 \epsilon=\frac45.
\end{equation}
Indeed,
\begin{equation}
 2\delta
 =1-\frac{\sqrt{21}}{11}
 <\frac45
 <\frac{100}{121}
 =4\delta(1-\delta),
\end{equation}
so Eq.~\eqref{eq:familywindow} applies and proves
\begin{equation}
 S(X;Y\Vert Z)=0
 <I(X;Y{\downarrow}Z)
 \le I_{\rm form}(X;Y\mid Z).
\end{equation}
Thus every bound-entangled state in the range \(3<\alpha\le4\) admits
measurements of a purification that produce a bound-information
distribution.  The same construction produces the identical classical
law from purifications of the separable members,
\(2\le\alpha\le3\).  These statements concern the specified classical
measurements of a purification.  They do not assert secrecy against an
adversary who instead retains arbitrary quantum side information.

\section{Further historical candidates}
\label{sec:historical-candidates}
Two further historical candidates for bound information remain
open.  For the second candidate of
Renner and Wolf~\cite{RennerWolf2003}, a key is distillable below the
conjectured transition $a=1/(4\sqrt2)$, while at and above it our
method provably cannot apply (see Appendix~\ref{sec:renwol-candidate}), so a vanishing rate needs a
different idea.  Example~2 of
Ref.~\cite{GisinWolf2000}, derived from Ref.~\cite{HorodeckiPPT1997}, is
likewise unresolved.

\end{document}